\documentclass[11pt]{article} 

\usepackage[utf8]{inputenc} 

\usepackage{geometry} 
\usepackage{graphicx} 

\usepackage{booktabs} 
\usepackage{array} 
\usepackage{paralist} 
\usepackage{verbatim} 
\usepackage{amsmath}
\usepackage{upgreek}
\usepackage{amssymb}
\usepackage{mathtools}
\usepackage{amsthm}
\usepackage[usenames, dvipsnames]{color}

\usepackage{subcaption}

\usepackage{algorithm}
\usepackage[noend]{algpseudocode}

\makeatletter
\def\BState{\State\hskip-\ALG@thistlm}
\makeatother
 
\usepackage{fancyhdr} 
\usepackage{sectsty}
\allsectionsfont{\sffamily\mdseries\upshape} 

\usepackage[nottoc,notlof,notlot]{tocbibind} 
\usepackage[titles,subfigure]{tocloft} 

\usepackage{graphicx}
\usepackage{epstopdf}
\DeclareGraphicsExtensions{.eps}

\usepackage{geometry}
\title{One-step Estimation of Networked Population Size with Anonymity Using Respondent-Driven Capture-Recapture and Hashing}

\usepackage{authblk}
\author[1,*]{Bilal Khan}
\author[1]{Hsuan-Wei Lee}
\author[1]{Kirk Dombrowski}

\affil[1]{Department of Sociology, University of Nebraska-Lincoln}
\affil[*]{Corresponding author: bkhan2@unl.edu}

\newtheorem{lemma}{Lemma}
\newtheorem{definition}{Definition}

\newcounter{example}[section]

\begin{document}
\date{}
\maketitle
\begin{abstract}

Estimates of population size for hidden and hard-to-reach individuals are of particular interest to health officials when health problems are concentrated in such populations. Efforts to derive these estimates are often frustrated by a range of factors including social stigma or an association with illegal activities that ordinarily preclude conventional survey strategies. This paper builds on and extends prior work that proposed a method to meet these challenges. Here we describe a rigorous formalization of a one-step, network-based population estimation procedure that can be employed under conditions of anonymity. The estimation procedure is designed to be implemented alongside currently accepted strategies for research with hidden populations. Simulation experiments are described that test the efficacy of the method across a range of implementation conditions and hidden population sizes. The results of these experiments show that reliable population estimates can be derived for hidden, networked population as large as 12{,}500 and perhaps larger for one family of random graphs. As such, the method shows potential for cost-effective implementation health and disease surveillance officials concerned with hidden populations. Limitations and future work are discussed in the concluding section.


\end{abstract}

\section{Introduction}

Population size estimation for hidden and hard-to-reach populations is of considerable interest to health officials seeking to prevent health problems that may be concentrated in such populations \cite{magnani_review_2005}, or when “reservoirs” of infection among a hidden population pose health threats to the ambient population in which the hidden population is embedded \cite{dombrowski_topological_2013, reluga_reservoir_2007}. In the former, treatable maladies can remain out of reach, multiplying eventual treatment costs when cases are discovered only in their most severe form. Such is the situation, for example, with mental illness among homeless and street dwelling populations \cite{bonin_typology_2009, burt_critical_1995}. In other situations, the “hidden” nature of a reservoir population may frustrate intervention efforts that are effective in the ambient population, preventing control of infections despite well-known contagion dynamics \cite{potterat_aids_1993}. A simple example, long-known to public health officials, is the high prevalence of sexually transmitted infections among commercial sex workers \cite{abdul2014estimating, law_spatial_2004, zohrabyan_determinants_2013}. Numerous other examples are recognized. In such situations, health officials seek to know the overall prevalence levels of maladies within a hidden population \textit{and} the size of those populations in order to understand the scope of treatment needs and overall social risk. 

Efforts to ascertain prevalence and size estimates are frustrated by a range of factors that produce the “hiddenness” of the population initially. Such factors include heavy social stigma that precludes a willingness on the part of members of the hidden population to reveal their membership. Such is the situation with people who inject drugs (PWID), who may be unwilling to self-identify as such under ordinary survey conditions \cite{darke_self-report_1998, harwood_sampling_2012}. Hiddenness due to stigma can be further compounded when such activities are illegal, when they carry heavy personal costs (such as when self-identified heterosexual men also have sex with other men), and when disease status is unknown (such as undiagnosed HIV infection rates among PWID). In these situations, conventional sampling is unreliable, and ordinary multiplier methods based on conventional sampling are rendered ineffective.

A number of strategies have been devised to address either the prevalence or population size (or both) aspects of this problem. These include capture-recapture \cite{larson_indirect_1994, vuylsteke_capturerecapture_2010}, chain referral \cite{biernacki_snowball_1981, platt_methods_2006}, venue-based \cite{haley_venue-based_2014, muhib_venue-based_2001} and cluster sampling \cite{burnham_mortality_2006}, and combinations of these. Among the most popular is respondent-driven sampling (RDS) \cite{heckathorn_extensions_2007, RDS2002, salganik_sampling_2004}, which has been adopted for use in many of the situations described above, and which is employed widely in HIV surveillance efforts both within the United States and beyond \cite{abdul2006implementation}. RDS employs an incentivized chain referral process to recruit a sample of the hidden population. Under restricted but recognized conditions, RDS can be shown to result in a steady-state, “equilibrium” sample. Numerous means have been derived for producing reasonable prevalence estimates from such a sample \cite{gile_respondent-driven_2010, gile_diagnostics_2015, mouw_network_2012,shi_model-based_2016, verdery_network_2015, wejnert_empirical_2009} while accounting for biases introduce in the referral process. The ease of implementing RDS, the fact that it can operate under conditions of anonymity (via number coupons that track referrals), and its rigorous treatment under a range of statistical modeling strategies have made it a popular choice for researchers working with hidden populations \cite{heckathorn_network_2017}. However, equally rigorous means for estimating the overall size of the hidden population from RDS derived data have been less successful—often resulting in widely varying estimates \cite{sulaberidze_population_2016}. Still, the ability of the RDS method to produce meaningful prevalence data remains, and presents considerable potential for use in size estimation.

Other efforts restricted to size estimation alone have been developed, including various versions of capture-recapture procedures (sometimes call mark-recapture procedures) \cite{domingo-salvany_analytical_1998, kruse_participatory_2003} and network scale-up methods (NSUM) \cite{bernard_counting_2010}. Capture-recapture efforts normally make use of a sample of the hidden population and some external, normally institutional knowledge-base (e.g. arrest records or hospital admissions) for estimation purposes \cite{hay_estimating_1996, vuylsteke_capturerecapture_2010}. In these cases, however, two assumptions must be met: (i) that the sample is representative of the hidden population more generally, and (ii) that everyone in the hidden population is equally likely to be “captured” in the official statistics used in the estimation \cite{jones_recapture_2014}. While representativeness can sometimes be assumed (as in the case of RDS), it is often difficult to establish the uniformity of the capture statistics, and often there are good reasons to believe that random capture is simply not the case. Frankly put, police arrests and hospital admission can seldom be assumed to draw randomly from the hidden population. Further, capture-recapture methods often require that the sample be identifiable in the institutional record, requiring that expectations on the part of sample respondents for anonymity be sacrificed. When working with hidden and highly stigmatized populations, such a sacrifice can be highly detrimental to both recruitment and informant reliability \cite{wolitski_effects_2009}.

Network scale-up methods are also used to establish the size of hidden populations, though work in this area remains at an early stage. Here members of the entire population (ambient plus hidden) are asked to report on the number of known associates who fit the hidden population criteria \cite{ezoe_population_2012, guo_estimating_2013}. This technique has the advantage of being employable under ordinary random sampling conditions that can make use of known sampling frames (i.e. mail surveys and/or random digit dialing) \cite{habecker_improving_2015}. However, this method assumes that ordinary people know whom among their associates fit the criteria for inclusion in the hidden category \cite{killworth_investigating_2006, salganik_assessing_2011}. This assumption raises suspicion in many of those situations in which we ordinarily wish to use it, as when we seek to estimate populations of PWID or sex workers. Under these conditions, individuals from the hidden population may not want their friends and associates to know about their membership in the categories, and may make efforts to hide this information. These efforts introduce “transmission” errors into NSUM estimates that are difficult to uncover or estimate. 

In previous work, we proposed a novel capture-recapture methodology for estimating the size of a hidden population from an RDS sample \cite{TELEFUNKEN2012}. Were such a result possible, it could easily be integrated into the conventional RDS framework, taking advantage of the wide body of work in that area and the ability of RDS to produce reasonable prevalence estimates. Our method was first proposed in several forms undertaken as quasi-experiments within actual data collection efforts with commercially sexually exploited children \cite{curtis_commercial_2008} and, later, users of methamphetamine \cite{wendel_dynamics_2011}. Both studies took place in New York City, and both made use of RDS samples. Subsequent implementation of the technique have lent further evidence of the effectiveness (and ease of implementation) of what we there referred to as the “telefunken” method. This method asks RDS sample respondents to report on others in the population known to them via an encoding of their associates telephone number and demographic features, avoiding the reliance on official statics or the need to draw two independent samples from the hidden population.\footnote{The technique was referred to as \textit{telefunken} because it entailed an encoding of the phone numbers of known associates in the hidden population. The code was created by taking a specified number of phone number digits, in order from last to first, and encoding each digit as 0/1 for even/odd, and again 0/1 for 0-4/5-9. This produced a binary code of length 2 x the number of phone digits specified in the protocol. This many-to-one encoding allowed for ongoing anonymity for both respondents and their reported associates, while enabling the matching of contacts across numerous respondent interviews. It also introduced the need to estimate the number of expected false matches created by the many-to-one encoding.} In essence, this “one-step” approach eases the assumptions normally associated with other capture-recapture methods, and can be accomplished via a single sample from the hidden population. If shown to be effective, this fact lends simplicity and greater cost-effectiveness to the size estimation procedure, potentially allowing for widespread application.

Given the interest in the technique \cite{merli_sampling_2016, mouw_network_2012, sulaberidze_population_2016}, this paper proposes a more rigorous formalization of a one-step, network-based population estimation procedure that can be employed under conditions of anonymity. In what follows we describe the technique and simulate its performance under a range of implementation conditions across a range of hidden population sizes. The simulations show considerable promise for the technique under the kinds of research scenarios normally associated with research among “hidden populations”. Limitations and further efforts toward validation/extension are discussed at the end of the paper.

As above, the framework for the paper assumes a population $V$ of size $|V| = n$. Under standard capture-recapture protocols, we identify $S\subset V$ to be a first uniform random ``capture'' sample, and $R\subset V$ be a second {\em independent} uniform random ``recapture'' sample. From independence assumptions, we know 
\begin{eqnarray}
\label{classical-prop}
\frac{|V|}{|S|} \approx \frac{|R|}{|S \cap R|}.
\end{eqnarray}
It follows that
\begin{eqnarray}
\label{lincoln}
|V| \approx \frac{|S| \cdot |R|}{|S \cap R|}.
\end{eqnarray}
This quantity is also known as the Lincoln-Peterson estimator \cite{LINCOLN1930,PETERSON1896}. In cases where $|S \cap R|=0$, the Chapman estimator \cite{CHAPMAN1951} is often applied
\begin{eqnarray}
\label{chapman}
|V| \approx \frac{(|S|+1) \cdot (|R|+1)}{|S \cap R|+1} - 1.
\end{eqnarray}

\section{Capture/Recapture on Graphs}

Where the Lincoln-Peterson technique assumed an unstructured population $V$, here we consider settings in
which $V$ has addition binary relational structure.  For the remainder of this paper, we will assume the population to 
be the ground set of a static undirected graph $G=(V,E)$.  Where appropriate, we will be explicit about any 
assumptions on the edge relation $E\subset V\times V$.

\subsection{The First Assay: Capture on Graphs}
\label{first-assay}

We replace the notion of a random ``capture'' sample, with a {\bf random respondent-driven ``capture'' sample}  \cite{RDS1997,RDS2002}.
A respondent-driven capture sample is a random variable $\text{RDS-CAPTURE}(G, s, c, n_0)$ requiring four parameters: an underlying networked population $G=(V,E)$, 
a specified number of seeds $s$, the number of coupons $c$ to be given to each subject, and $n_0$ the target capture size.  
Informally stated, the procedure chooses $s$ random initial ``seed'' subjects in the network.  Each of these subjects is asked to 
participate in a ``referral'' process by being given $c$ coupons to
be distribute among their peers.  When those peers come in for interview, they too in turn, are given $c$ coupons and participate in the
referral process.  The scheme proceeds recursively in this manner until $n_0$ individuals have been recruited and interviewed.
If and whenever the referral process stalls before $n_0$ subjects have been interviewed, a new seed is recruited.
Participation incentives are arranged to ensure that no subject will be the recipient of more than one coupon.
Note that this breadth-first search process always yields a collection of disjoint trees \cite{Bollobas98a}.

The formal description of this real-world process in a simulated setting is given in the RDS-CAPTURE procedure 
presented as Algorithm \ref{rds-algo} (pp. \pageref{rds-algo}).

\begin{algorithm}
\caption{random respondent-driven ``capture'' sample}\label{rds-algo}
\begin{algorithmic}[1]
\Procedure{RDS-CAPTURE}{$G$, $s$, $c$, $n_0$}
\State $t \gets 0$
\State $S_0 \gets \{ v_1, \ldots, v_s\}$ a set of $s$ distinct ``seeds'' uniformly at random from $V[G]$.
\State $T_0 \gets \emptyset$.
\State $F_0 \gets S_0$.
\Repeat 
\State $t \gets t+1$
\State $x_{t} \gets $ a uniformly randomly chosen element from $F_{t-1}$
\State $N(x_{t}) \gets \{ v\in V[G] \setminus S_{t-1} \;|\; (x_{t},v) \in E[G] \}$ its undiscovered neighbors
\If {$|N(x_{t})| \leqslant c$}
\State $R(x_{t}) \gets N(x_{t})$
\Else
\State $R(x_{t}) \gets$ a uniformly  random chosen size-$c$ subset of $N(x_{t})$
\EndIf
\State $S_{t} \gets S_{t-1} \cup \{ x_{t} \} \cup R(x_{t})$
\State $T_{t} \gets T_{t-1} \cup \{ (x_t, v) \;|\; v \in R(x_{t})\}$
\State $F_{t} \gets F_{t-1} \setminus \{ x_{t} \} \cup R(x_{t})$
\If {$F_{t} = \emptyset$ and $|S_{t}| < n_0$}
\State $F_{t} \gets \{ v \}$ a single ``seed'' chosen at random from $V[G] \setminus S_{t}$
\EndIf
\Until {$|S_{t}| \geqslant n_0$}
\State \Return $(S_t, T_t)$
\EndProcedure
\end{algorithmic}
\end{algorithm}

\subsection{The Second Assay: Recapture on Graphs}
\label{second-assay}

In classical capture-recapture methods, the independence of the two assays is essential to the estimation procedure.
Here we will abandon this requirement, and instead derive the recapture set from the capture set, via a mechanical (albeit randomized) procedure.
We show later, empirically, that this definition of recapture set can be used as the foundation of reasonable population size estimates for large families of
networked populations.  What is being leveraged in this estimation strategy is not the independence of the first and second assays (they are in fact
not independent)--rather, the technique rests on the intrinsic geometry of the network structures themselves.

A recapture sample is a random variable $\text{RECAPTURE}(G, (S,T), p)$ requiring three parameters: the underlying networked population $G=(V,E)$, 
a subgraph representing the first assay RDS-CAPTURE sample $(S,T)$, and the maximum number of reports per subject $p$.  
Informally stated, the procedure interviews each of the individuals in the first assay $S$ and asks them to reveal the identities of upto $p$ ``reports''.
These reports may be any individuals in the subjects' ego network {\em except} their recruiter or recruitees within the first assay; 
note that an individual may report members of the first assay who are not their immediate recruiter or recruitees.
Stated alternately, each individual $v$ is asked to report on upto $p$ of their neighbors in $G=(V,E)$ who were {\em not} $v$'s neighbor in the RDS tree $(S,T)$.
These reports, obtained from each of the individuals in $S$, are then amalgamated into a single {\em multiset}, taking care to 
retain information about multiplicities (since the same individual may be reported by multiple subjects in the first capture).

The formal description of the RECAPTURE procedure is given in 
the pseudocode of Algorithm \ref{recap-algo} (pp. \pageref{recap-algo}).

\begin{algorithm}
\caption{random respondent-driven ``recapture'' sample}\label{recap-algo}
\begin{algorithmic}[1]
\Procedure{RECAPTURE}{$G$, $(S, T)$, $p$}
\State $R \gets \emptyset$
\ForAll {$v \in S$} 
\State $N_T(v) \gets \{ v\in S  \;|\; (x_{t},v) \in T \}$ are $v$'s neighbors in $(S,T)$
\State $N_G(v) \gets \{ v\in V[G]  \;|\; (x_{t},v) \in E[G] \}$ are $v$'s neighbors in $G$
\State $C_v \gets N_G(v) \setminus N_T(v)$ are $v$'s candidate reports
\If {$|C_v| \leqslant p$}
\State $R_v \gets C_v$ 
\Else
\State $R_v \gets$ a uniformly  random chosen size-$p$ subset of $ C_v$
\EndIf
\State $R \gets R \uplus R_v$ 
\EndFor
\State \Return $R$
\EndProcedure
\end{algorithmic}
\end{algorithm}

In what follows we will be presenting a series of population size estimators based on the set-valued procedure RDS-CAPTURE and the multiset-valued procedure RECAPTURE.
To be able to exposit these new estimators precisely, it is necessary to introduce some formal notations concerning multisets and their properties.

\begin{definition}
Given a universal set $\mathcal{U}$, a {\bf multiset} $A$ (in $\mathcal{U}$)  is defined as a function $\chi_{A}: \mathcal{U} \rightarrow \mathbb{Z}^{*}$, where each element $a \in  \mathcal{U}$ is said to ``appear'' $\chi(a)$ times in $A$.  The concept of multiset generalizes the concept of set; the latter
is subject to the additional restriction $\forall a\in A$, $\chi(a) \in \{0,1\}$.

Given a multiset $A$, we define $A^* = \{ a \in A \;|\; \chi_A(a) > 0 \}$; note that $A^*$ is always a set.  We define $|A| = |A^*|$ 
and $\langle A \rangle = \sum_{x \in \mathcal{U}} \chi_{A}(x)$, the cardinality of a set and a multiset, respectively.  If $A$ and $B$ are multisets, then the multisets $A\cup B$, $A\cap B$, $A\uplus B$, and $A\setminus B$ are 
defined by taking 
\begin{eqnarray*}
\chi_{A\cup B}(x) &=& \max \;\{\; \chi_{A}(x),\chi_{B}(x) \;\}\; \\
\chi_{A\cap B}(x) &=& \min  \;\{\; \chi_{A}(x),\chi_{B}(x) \;\}\; \\
\chi_{A\uplus B}(x) &=& \chi_{A}(x) + \chi_{B}(x)\\
\chi_{A\setminus B}(x) &=& \max  \;\{\; 0, \chi_{A}(x) - \chi_{B}(x) \;\} ,
\end{eqnarray*}
for each $x\in \mathcal{U}$.  In addition, we define the ``filtering of $A$ by $B$'' by taking 
  \[
   \chi_{A\big\rvert {B}}(x) = \left\{
                \begin{array}{ll}
                  \chi_{A}(x) & \text{if  } \chi_{B}(x) > 0\\
                  0               &\text{if }  \chi_{B}(x) = 0
                \end{array}
              \right.
  \]
Note that if $A$ and $B$ are sets, then $A\big\rvert B \equiv A \cap B$; filtering may thus be viewed as a form of set intersection generalized to multisets.
\end{definition}

The processes of RDS-CAPTURE and RECAPTURE are illustrated in Figure \ref{caprecap-concept}.  The underlying graph $G=(V,E)$ is shown on the left.  On the right, we start from $s=1$ seed, giving each subject $c=2$ coupons to distribute.  The coupons flow along the red directed edges shown in the figure on the right, and the order in which subjects come in for their first interview is indicated by their label.  The process continues until our target sample size $n_0=7$ is reached.  When subjects return for their second interview (to collect referral incentives), they are asked to provide up to $p=5$ ``reports'' about their peers (other than their referrer and referees).  These are shown as directed green edges in the figure.  Note that (i) some individuals reported are in the capture sample (e.g. 2 is reported by $5$, and both were members of the first assay); (ii) some individuals are reported multiple times (e.g. A was reported by both 2 and 5); (iii) some individuals reported are outside the first assay (e.g. D was reported by 1, but D was not part of the first assay; (iv) some individuals report nobody (e.g. 3 reports nobody since all its peers were either its referrer or its referees.  In this scenario, the capture set  $S=\{1,2,3,4,5,6,7\}$, and the recapture multiset $rS=\{C,D,5,A,A,2,B,C,7,6,E\}.$\footnote{$1$ reports $C,D$; $2$ reports $5$; $3$ reports nobody; $4$ reports $A$; $5$ reports $2,A,B,C$; $6$ reports $7$; $7$ reports $6,E$.}

\begin{figure}[h!]
\centering
\includegraphics[width = 3.0in]{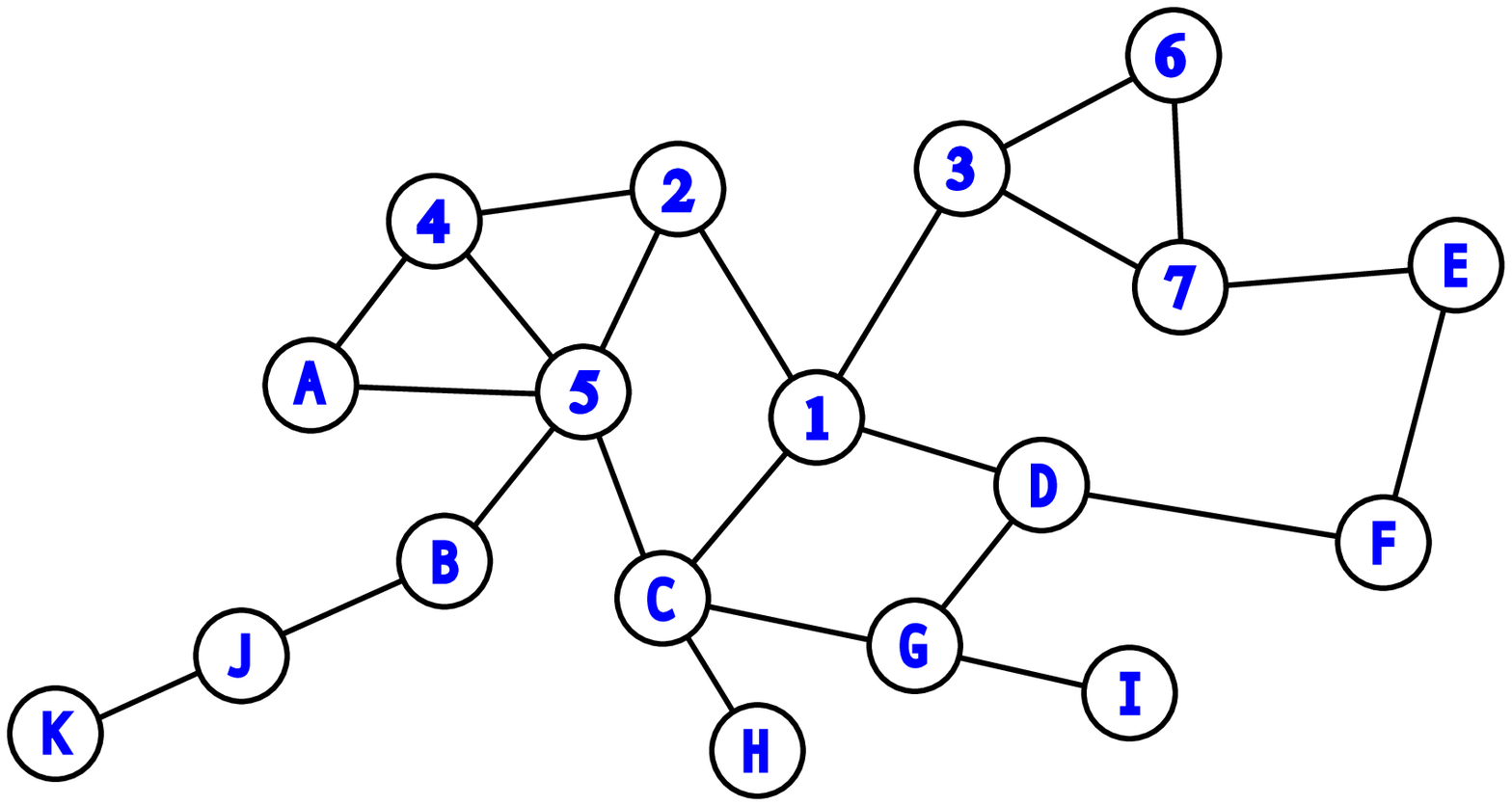}
\includegraphics[width = 3.0in]{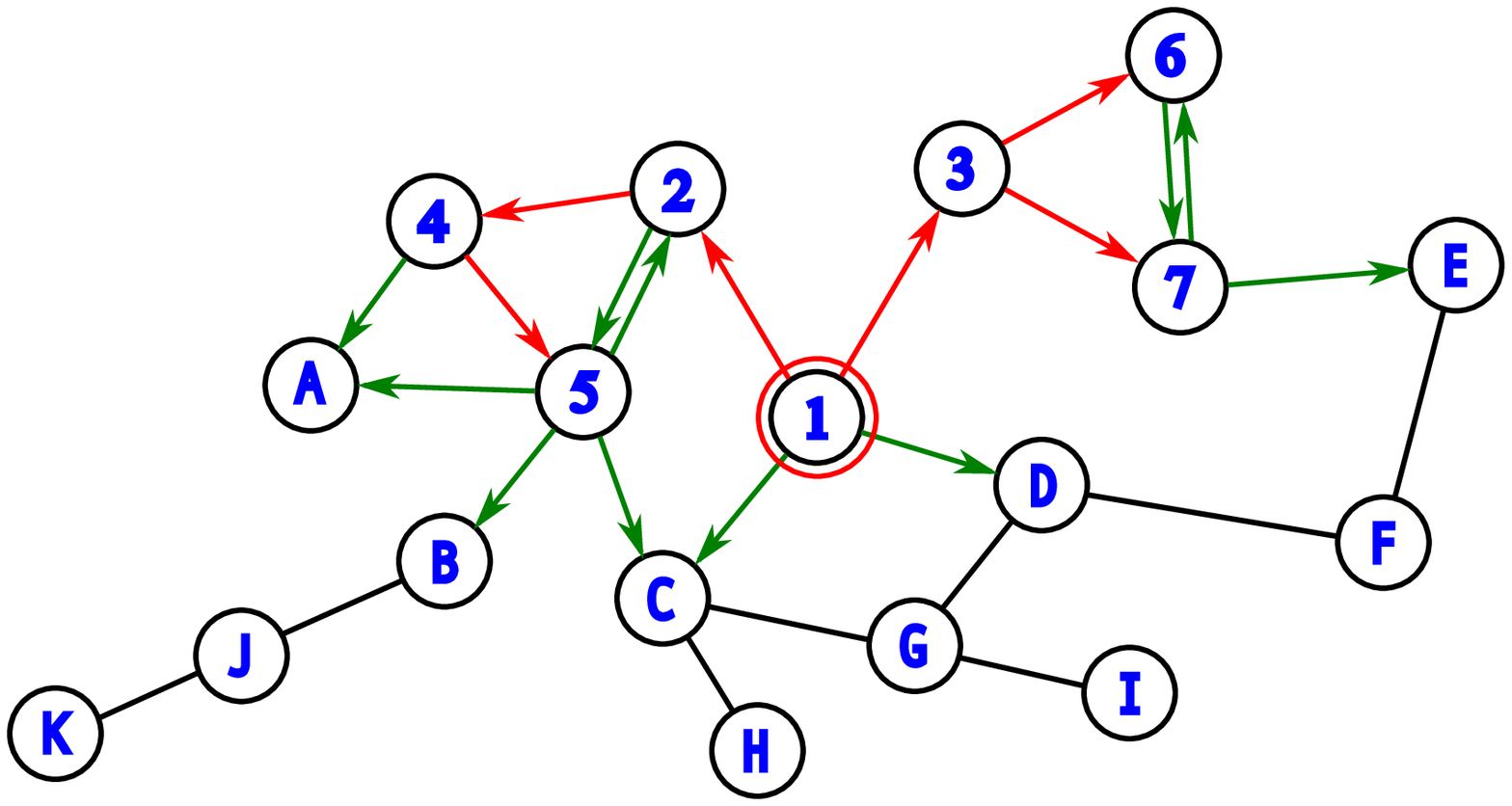}
\caption{RDS-CAPTURE (red edges) and RECAPTURE (green edges)}
\label{caprecap-concept}
\end{figure}
\section{New Estimators}

In light of the previous section, we see that given a graph $G=(V,E)$ with $|V|=n$, and four fixed parameters:
{\em
\begin{enumerate}
  \setlength{\itemindent}{0.5in}
\item Capture size $n_0$. 
\item Number of seeds $s$. 
\item Number of coupons $c$ to be given to each subject.
\item Number of reports per subject $p$.
\end{enumerate}
}
 (all of which are natural numbers), we can 
\begin{itemize}
  \setlength{\itemindent}{0.5in}
\item Sample a random capture set $(S,T) \coloneqq \text{RDS-CAPTURE}(G, s, c, n_0)$, and then 
\item Obtain a recapture set $rS \coloneqq \text{RECAPTURE}(G, (S,T) , p)$.  
\end{itemize}

As will be presented in later sections, we have verified empirically that for large families of graphs $G$
\begin{eqnarray}
\label{est1}
\frac{|V|}{|S|} \approx \frac{ |(rS)^*| }{| (rS \big| S)^* |}
\end{eqnarray}
Note that  expression (\ref{est1}) mirrors the classical proportionality expressed in Eqn. (\ref{classical-prop}).
The set $rS \big| S$ will be denoted $M(rS, S)$ as it is the set of ``matches'' between the reports collected 
(in the second assay/recapture) and the original subjects (of the first assay/capture).
This is formalized for general use in the next definition.
\begin{definition}
Let $A, B$ be multisets from the universe $\mathcal{U}$. Then the set of {\bf matches} of $A$ in $B$ is denoted $M(A,B) \coloneqq A\big\rvert B$.
\end{definition}

The empirically verified proportionality (\ref{est1})  leads to our first estimator:
\begin{eqnarray}
\label{n1estimation}
n_{1} = \frac{|S| \cdot  |(rS)^*| }{| M(rS, S)^*|}.
\end{eqnarray}

A procedural description of our third population estimator defined in Eqn. (\ref{n1estimation}) is given in the pseudocode of Algorithm \ref{rdscaprecap-algo} (pp. \pageref{rdscaprecap-algo}).

\begin{algorithm}
\caption{Estimation using RDS Capture/Recapture}\label{rdscaprecap-algo}
\begin{algorithmic}[1]
\Procedure{HASHING-ESTIMATE}{$G=(V,E)$, $s$, $c$ , $n_0$}
\State $(S,T) \gets $RDS-CAPTURE$(G, s, c, n_0)$ or collected by survey
\State $rS \gets $RECAPTURE$(G, (S,T), p)$ or collected by survey
\State \Return $\frac{|S| \cdot  |(rS)^*| }{| M(rS, S)^* |}$
\EndProcedure
\end{algorithmic}
\end{algorithm}

Estimator (\ref{n1estimation}) works well in practice, but the rationale for its efficacy is necessarily distinct from the 
reasoning underlying the classical Lincoln-Petersen estimator.  In particular, the set $S$ and multiset 
$rS$ do {\em not} constitute independent assays because of their structural relationship within the 
graph $G$.  In other words, for all $u,v \in V$ it is generally the case that
\begin{eqnarray}
\label{bad-not-independent}
Pr(v\in rS \:|\; u\in S) \neq Pr(v\in rS \:|\; u\not\in S).
\end{eqnarray}
Informally stated, knowing whether or not $u$ was placed in $S$ can impact the probability that $v$ will be placed in $rS$.  For example, if $(u,v) \in E$ then the left hand-side of the above inequality (\ref{bad-not-independent}) always evaluates to $1$; if (additionally) $v$ has degree $1$ in $G$, then the right-hand side always evaluates to $0$.
Indeed, as we shall see, estimator (\ref{n1estimation}) works well in practice only for certain families of graphs.  Fortunately, many of these families are ubiquitous
in social organization.

\subsection{Anonymity via Hashing}

Significant obstacles arise in the direct application of estimator (\ref{n1estimation}) in practice.  In our field work, we frequently seek to measure the sizes of stigmatized networked populations (e.g. people who inject drugs, sex workers, criminal elements, etc.) Individuals within these social communities naturally seek to remain ``hidden'', and thus the membership of sets $S$ and $rS$ is often not explicitly knowable because individuals are reluctant to unambiguously identify themselves or their social network peers.  Nevertheless, simulation experiments show that estimator (\ref{n1estimation}) is effective in large classes of networked populations, so in settings where anonymity is not required, this estimator may be practically applied. 

To begin to address questions of population estimation under the requirement of anonymity, we introduce a ``coding'' or {\bf hashing}  \cite{CARTER1979} function $\uppsi:V\rightarrow H$ which provides anonymity to members of the  population.  In practice, $\uppsi(v)$ might be an obtained by amalgamating a well-defined tuple of characteristics of $v$ which are known to $v$'s friends (e.g. $v$'s gender, phone number, hair color, approximate age, racial category, etc.)   A related coding technique was used in our earlier work on estimating the size of the methamphetamine user population in New York City, where it was referred to as the ``telefunken'' code \cite{TELEFUNKEN2012}.  Central to hashing is that it is many-to-one, and hence not readily invertible.  For our purpose, what is important is the following well-known property of universal hash functions:  For all $v \in V$ and $x \in H$, the probability $Pr(\uppsi(v) = x) \approx 1/|H|$.   Supporting anonymity in this manner introduces a new additional 5th parameter:
{\em 
\begin{enumerate}
  \setlength{\itemindent}{0.5in}
  \setcounter{enumi}{4}
\item Hash space size $|H|$.
\end{enumerate}
}
Let $\uppsi$ be the random hashing function from $V$ to $H$; note that $\uppsi$ takes sets (and multisets) in $V$ to multisets in $H$. Under the 
action of $\uppsi$, we have $\uppsi S$, $\uppsi rS$, and $M( \uppsi rS, \uppsi S)$ which are all multisets in $H$.  These sets are knowable even in 
settings where anonymity is required, since $\uppsi$ is assumed to be a random function, and hence difficult to invert.

The action of $\uppsi$ on the population $V$ and the associated sets and multisets are illustrated in Figure \ref{psi-concept}.

\begin{figure}[h]
\centering
\includegraphics[width = 4.0in,scale=1]{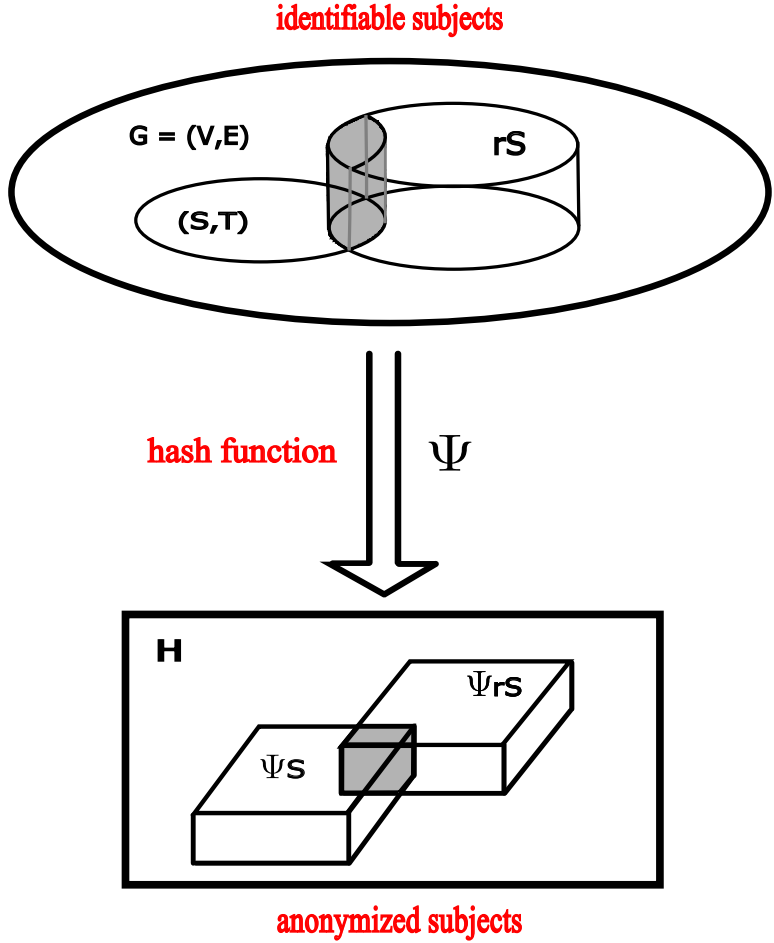}
\caption{The mapping of $\uppsi$ of $V[G]$ into $H$. In space $G$, $S$ is a set and $rS$ is a multiset. Under the mapping of the hash function $\uppsi$, their images $\uppsi S$ and $\uppsi rS$ are both multisets in the hash space $H$. }
\label{psi-concept}
\end{figure}

We may now consider the following estimator as a natural analogue for (\ref{n1estimation}) in the presence of coding function $\uppsi$ for population size:
\begin{eqnarray}
\label{n2estimation}
n_{2} = \frac{\langle \uppsi S \rangle \cdot \langle \uppsi rS \rangle}{\langle M(\uppsi rS, \uppsi S) \rangle}
\end{eqnarray}
Since 
\begin{eqnarray*}
\langle \uppsi S \rangle &=& |S|\\
\langle M(\uppsi rS, \uppsi S)\rangle &\geqslant& \langle \uppsi M( rS , S) \rangle = \langle M(rS, S) \rangle \geqslant | (rS \big| S)^* |\\
\langle \uppsi rS \rangle &=& \langle rS \rangle \geqslant |(rS)^*|
\end{eqnarray*}
we know that $n_2$ and $n_1$ may diverge.  In seeking to improve estimator (\ref{n2estimation}), we have two possible opportunities---corresponding to the two inequalities seen above. 
We will pursue each successively, in turn.

First, 
some of elements in $M(\uppsi rS, \uppsi S)$
arise simply because the large multisets $rS$ and set $S$ are being randomly 
hashed to a finite set $H$, and not because $rS$ and $S$ share a non-empty intersection.  To quantify how many such ``false'' matches are occurring by sheer chance, 
we need to introduce some formal notation.

\begin{definition}
Let $A$ be a set and $B$ be a multiset on the universe $\mathcal{U}$. Let $\uppsi: \mathcal{U} \rightarrow H$ be a function uniformly random distributes each element $u \in \mathcal{U}$ to $H$. Then the set of {\bf false matches} of $B$ in $A$ is denoted 
 $$F_{\uppsi}(A,B) := (\uppsi A \cap \uppsi B) \backslash \uppsi M(B,A).$$
\end{definition}

\begin{lemma}
\label{baseline}
Let $A$ be a set and $B$ be a multiset on the universe $\mathcal{U}$ with $A \cap B = \emptyset$. Let $\uppsi: \mathcal{U} \rightarrow H$ be a random function which sends each element $u \in \mathcal{U}$ to a uniformly random element $\uppsi(u) \in H$. The expected number of false matches is given by
\begin{eqnarray}
\label{E(f)}
E [\langle F_{\uppsi}(A,B) \rangle] = \frac{\langle B \rangle}{|B^{*}|} \cdot \sum^{\min \{a, b, |H|\}}_{k = 0} k \cdot {|A| \choose k} {|B^{*}| \choose k} \Big(\frac{k}{|H|}\Big)^{2k} \big(\frac{|H|-k}{|H|}\Big)^{|A|+|B^{*}|-2k}
\end{eqnarray}
\end{lemma}

\begin{proof}
First note that $|A| = \langle \uppsi A \rangle$ and $\langle B \rangle = \langle \uppsi B \rangle$. In the set A, there are $|A|$ elements with cardinality~1. In the set $B$, we may assume there are $|B^{*}|$ elements and each has a cardinality $\langle B\rangle/|B^{*}|$. 

Given $\uppsi$ is a random function, the image of $A$ and $B$ under $\uppsi$ are (set-valued) random variables, and their set cardinalities in the target intersection set is also a random variables.

For any target set $T\subset H$.  Let $X_A$ (respectively, $X_{B^*}$) denote the number of elements of $A$ (respectively, ${B^*}$) which are mapped by $\uppsi$ into $T := (\uppsi A \cap \uppsi B)^{*}$ .   If $|(\uppsi A \cap \uppsi B)^{*})|=k$, then

\begin{align*} 
P\big(X_{A} = k \big) &= {|A| \choose k}\Big(\frac{k}{|H|}\Big)^{k} \big(\frac{|H|-k}{|H|}\Big)^{|A|-k} \\ 
P\big(X_{B^{*}} = k \big) &= {|B^{*}| \choose k}\Big( \frac{k}{|H|}\Big)^{k} \big(\frac{|H|-k}{|H|}\Big)^{|B^{*}|-k}
\end{align*}

Given $A \cap B = \emptyset$, we want to quantify the cardinality of multiset $\uppsi B \big\rvert \uppsi A$, we focus on the set 
$$C := B^{*} \cap \uppsi^{-1}((\uppsi A)^{*}\cap (\uppsi B)^{*})$$  
Summing multiplicity of each element in set $C$ gives us the cardinality of multiset $\uppsi B \big\rvert \uppsi A$. Our estimate of the expected 
cardinality of the false match multiset is then expressible as:
\begin{equation*}
  \begin{split}
   E[\langle F_{\uppsi}(A,B) \rangle]
   &= E \bigg ( \sum_{y \in C} \chi_{B} (y) \bigg )
  \end{split}
 \end{equation*}
Let $m' := \min \{|A|, |B^{*}|, |H|\}$.  The right-hand quantity above can then be further explicitly described as follows:
\begin{equation*}
  \begin{split}
   E \bigg ( \sum_{y \in C} \chi_{B} (y) \bigg )
   &=  \frac{\langle B \rangle}{|B^{*}|} \cdot \sum^{m'}_{k = 0} k \cdot P\big(X_{A} = k \big) \cdot P\big(X_{B^{*}} = k \big)\\
   &= \frac{\langle B \rangle}{|B^{*}|} \cdot \sum^{m'}_{k = 0} k \cdot {|A| \choose k} {|B^{*}| \choose k} \Big(\frac{k}{|H|}\Big)^{2k} \big(\frac{|H|-k}{|H|}\Big)^{|A|+|B^{*}|-2k}.
  \end{split}
 \end{equation*}
The lemma is proved.
\end{proof}

Returning to the task of improving estimator (\ref{n2estimation}), we see that one can now replace the quantity 
$\langle M(\uppsi rS, \uppsi rS)\rangle$ used as the denominator in estimate $n_2$, with $\langle M(\uppsi rS, \uppsi rS)\rangle - E[\langle F_{\uppsi}(S, rS) \rangle]$, since by Lemma \ref{baseline} 
we expected $E[\langle F_{\uppsi}(S, rS) \rangle]$ matches even when {\em disjoint} sets having size/mass equivalent to that of $S$ and $rS$ are 
mapped to a coding space of size $|H| = m$.  

The second idea in improving estimator (\ref{n2estimation})
comes from realizing that although $\langle \uppsi rS \rangle$ used in the numerator of estimate $n_2$ is greater than $|(rS)^*|$ used in the numerator of $n_1$, one may be able to derive an approximation
for the latter quantity.
The following combinatorial observation is useful here: If $|(rS)^*|$ balls are cast randomly into $|H|$ boxes, then the expected number of
empty boxes is readily seen to be
$$|H|\cdot\left( \frac{|H|-1}{|H|} \right)^{|(rS)^*|}$$

Although we do not know $|(rS)^*|$, we do know that when this set was mapped by $\uppsi$ into $H$,
its image avoided $|H| - | (\uppsi rS)^*|$ elements of $H$.  Assuming the mean outcome, we may estimate $|(rS)^*|$
by solving for it in the equation:
\begin{eqnarray*}
|H| - | (\uppsi rS)^*| & = & |H|\cdot\left( \frac{|H|-1}{|H|} \right)^{|(rS)^*|}
\end{eqnarray*}
Such an analysis leads us to estimate
$$
|(rS)^*| \approx \frac{\log \left(1 - \frac{| (\uppsi rS)^*|}{|H|}\right)}{\log \left(1 - \frac{1}{|H|}\right)}.
$$

Combining the previous two conclusions, we get our third and final population estimate
\begin{eqnarray}
\label{n3estimation}
n_{3} = \frac{\langle \uppsi S \rangle}{\langle M(\uppsi rS, \uppsi S) \rangle - E[\langle F_{\uppsi}(S, rS) \rangle]} \cdot \frac{\log \left(1 - \frac{| (\uppsi rS)^*|}{|H|}\right)}{\log \left(1 - \frac{1}{|H|}\right)}
\end{eqnarray}
where $E[\langle F_{\uppsi}(S, rS) \rangle]$ is as defined in Eqn. (\ref{E(f)}).

A procedural description of our third population estimator defined in Eqn. (\ref{n3estimation}) is given in the pseudocode of Algorithm \ref{hashing-algo} (pp. \pageref{hashing-algo}).

\begin{algorithm}
\caption{Estimation using RDS Capture/Recapture {\em with Anonymity using Hashing}}\label{hashing-algo}
\begin{algorithmic}[1]
\Procedure{HASHING-ESTIMATE}{$G=(V,E)$, $s$, $c$ , $n_0$, $p$, $m$}
\State $(S,T) \gets $RDS-CAPTURE$(G, s, c, n_0)$ or collected by survey
\State $rS \gets $RECAPTURE$(G, (S,T), p)$ or collected by survey
\State $\uppsi \gets $ a function from $V \rightarrow \{1, 2, \ldots, m\}$, random or collected by survey
\State $\uppsi S \gets \uppsi(S)$
\State $\uppsi rS \gets \uppsi(rS)$
\State \Return $\frac{\langle \uppsi S \rangle}{\langle M(\uppsi rS, \uppsi S) \rangle - E[\langle F_{\uppsi}(S, rS) \rangle]} \cdot \frac{\log \left(1 - \frac{| (\uppsi rS)^*|}{|H|}\right)}{\log \left(1 - \frac{1}{|H|}\right)}$
\EndProcedure
\end{algorithmic}
\end{algorithm}

\subsection{Avoiding Pathologies via Bootstrap}
In practice, several difficulties arise in applying the $n_3$ estimator (\ref{n3estimation}).  The most significant of these
is that the denominator quantity can become negative if
$$
E[\langle F_{\uppsi}(S, rS) \rangle] > \langle M(\uppsi rS, \uppsi S) \rangle.
$$
Such situations can arise since $E[\langle F_{\uppsi}(S, rS) \rangle]$ is an expectation,
 but for a specific choice of $\uppsi$,  $M(\uppsi rS, \uppsi S)$ may be lower than this
 expectation.  This leads to the embarrassing possibility of negative population estimates.
 At the same time, we note that over large numbers of trials, the $\langle M(\uppsi rS, \uppsi S) \rangle - E[\langle F_{\uppsi}(S, rS) \rangle]$ is on average always positive.  The problem in practice of course, is that one cannot afford to conduct multiple trials.  Sampling a population and interviewing $n_0$ subjects
 are both time consuming and expensive.  If at he end of such a process, we were to get a negative population estimate, all our efforts would be in vain.  How can we address the situation? 
 
Our approach here is to carry out a bootstrap procedure.  Informally, suppose that we have completed our survey, and as describe above, obtained our RDS-based capture set $\uppsi S$ and recapture multiset $\uppsi rS$.  We now simulate artificial scenarios in which only $\alpha$ fraction of the subjects in $S$ were discovered ($0 \leqslant \alpha \leqslant 1$), along the same referral tree as in our actual survey.  This is a thought experiment, and does not require additional sampling of the population.  The thought experiment yields new subset $\uppsi S' \subseteq \uppsi S$ and submultiset $\uppsi rS' \subseteq \uppsi rS$.  We compute the quantity $\langle M(\uppsi rS', \uppsi S') \rangle - E[\langle F_{\uppsi}(S', rS') \rangle]$ and if this is positive, accumulate it into a multiset $D$ of denominator estimates.  We then use the average Ave($D$) as the denominator, in place of $\langle M(\uppsi rS', \uppsi S') \rangle - E[\langle F_{\uppsi}(S', rS') \rangle]$, when computing the $n_3$ estimator (\ref{n3estimation}).  A procedural description of this ``bootstrapped'' version of our third population estimator is given in the pseudocode of Algorithm \ref{bootstrap-algo} (pp. \pageref{bootstrap-algo}).

\begin{algorithm}
\caption{{\em Bootstrapped} Estimation using RDS Capture/Recapture with Anonymity using Hashing}\label{bootstrap-algo}
\begin{algorithmic}[1]
\Procedure{BOOTSTRAPPED-ESTIMATE}{$G=(V,E)$, $s$, $c$ , $n_0$, $p$, $m$, $\alpha$, $\kappa$}
\State $(S,T) \gets $RDS-CAPTURE$(G, s, c, n_0)$ or collected by survey
\State $rS \gets $RECAPTURE$(G, (S,T), p)$ or collected by survey
\State $\uppsi \gets $ a function from $V \rightarrow \{1, 2, \ldots, m\}$, random or collected by survey
\State $\uppsi S \gets \uppsi(S)$
\State $\uppsi rS \gets \uppsi(rS)$
\State $D \gets \emptyset$
\For {$i \gets 1 \ldots \kappa$} 
\State $(S',T') \gets $RDS-CAPTURE$((S,T), s, c, \lceil \alpha n_0 \rceil)$
\State $rS' \gets $RECAPTURE$(G, (S',T'), p)$
\State $\uppsi S' \gets \uppsi(S')$
\State $\uppsi rS' \gets \uppsi(rS')$
\State $d_i \gets \langle M(\uppsi rS', \uppsi S') \rangle - E[\langle F_{\uppsi}(S', rS') \rangle]$
\State $D \gets D \uplus \{d_i\}$
\EndFor
\State $M \gets $Ave$(D)$
\State \Return $\frac{\langle \uppsi S \rangle}{{M}} \cdot \frac{\log \left(1 - \frac{| (\uppsi rS)^*|}{m}\right)}{\log \left(1 - \frac{1}{m}\right)}$
\EndProcedure
\end{algorithmic}
\end{algorithm}

\nocite{STIRLING1988}

\section{Empirical Results}

The results presented in this section are experimental in nature.  We would like to evaluate the extent to which
our estimates $n_1$ and $n_3$, presented in Eqns. (\ref{n1estimation}) and (\ref{n3estimation}) respectively, are able to predict the size of the network from which the sample is drawn.  In the graphs
that follow, we refer to the $n_1$ estimate as
``RDS full-knowledge'' and the $n_3$ estimate as ``RDS + ANON/hashing''.

As a {\bf baseline} scenario, we take the following as our parameter settings.
{\em
\begin{enumerate}
  \setlength{\itemindent}{0.5in}
\item Capture size $n_0=500$. 
\item Number of seeds $s=6$. 
\item Number of coupons $c=3$ to be given to each subject.
\item Number of reports per subject $p=25$.
\item Hash space size $|H|=m=3{,}125$.
\end{enumerate}
}

The rationale for these values is derived from standard implementation of RDS in health research such as the National HIV Behavioral Survey (Injection Drug User round), where samples of 500 current users are recruited via RDS using 5-10 seeds and 3 referral coupons for each respondent. Similar measures have been adopted internationally for surveillance and estimation efforts on related key populations.

In a series of 5 experiments, we vary each of the above 5 parameters, one at a time, while keeping the other 4 parameters 
fixed at their baseline values.  Each of these series of experiments iss carried out 100 times on each of 
10 Barabasi-Albert networks of size 6{,}250, 12{,}500, 25{,}000, and 50{,}000 nodes.
In all experiments, the quantity $E[\langle F_{\uppsi}(S, rS) \rangle]$ appearing in the definition of $n_3$ and Eqn. (\ref{n3estimation}) 
is estimated via Monte-Carlo bootstrapping techniques described above rather than by explicit exact computation
using the closed-form expression (\ref{E(f)}).

Figure \ref{rds-graphs} shows how population estimates change as the RDS {\bf capture sample size} is varied from 
200 subjects to 1000 subjects inside networked populations
 comprised of (a) 6{,}250, (b) 12{,}500, (c) 25{,}000, and (d) 50{,}000 individuals.  In these experiments,
the number of seeds $s=6$, the number of coupons given to each subject $c=3$,
the number of reports provided by each subject is $p=25$, and the hash space size $|H|=3{,}125$.
These results show that, for hidden populations ranging in size from 6{,}250 to from 50{,}000, the method described above produces meaningful estimates with moderate variance when the hashing of respondent identity does not take place. This is especially true for RDS sample sizes ranging from 500 to 1000 (which is in line with other expectations for RDS design effects). Where respondent identities are hashed (and anonymity maintained) RDS samples of 500 or more produce small over-estimations for networks of up to 25{,}000 individuals. In populations larger than 25{,}000, the hashed results show extremely large variance regardless of the RDS sample size (up to an experimental sample size of 1000). 

\begin{figure}[p]
\centering
\begin{tabular}{cccc}
\subcaptionbox{Actual population size 6{,}250\label{rds6}}{\includegraphics[width = 3.2in]{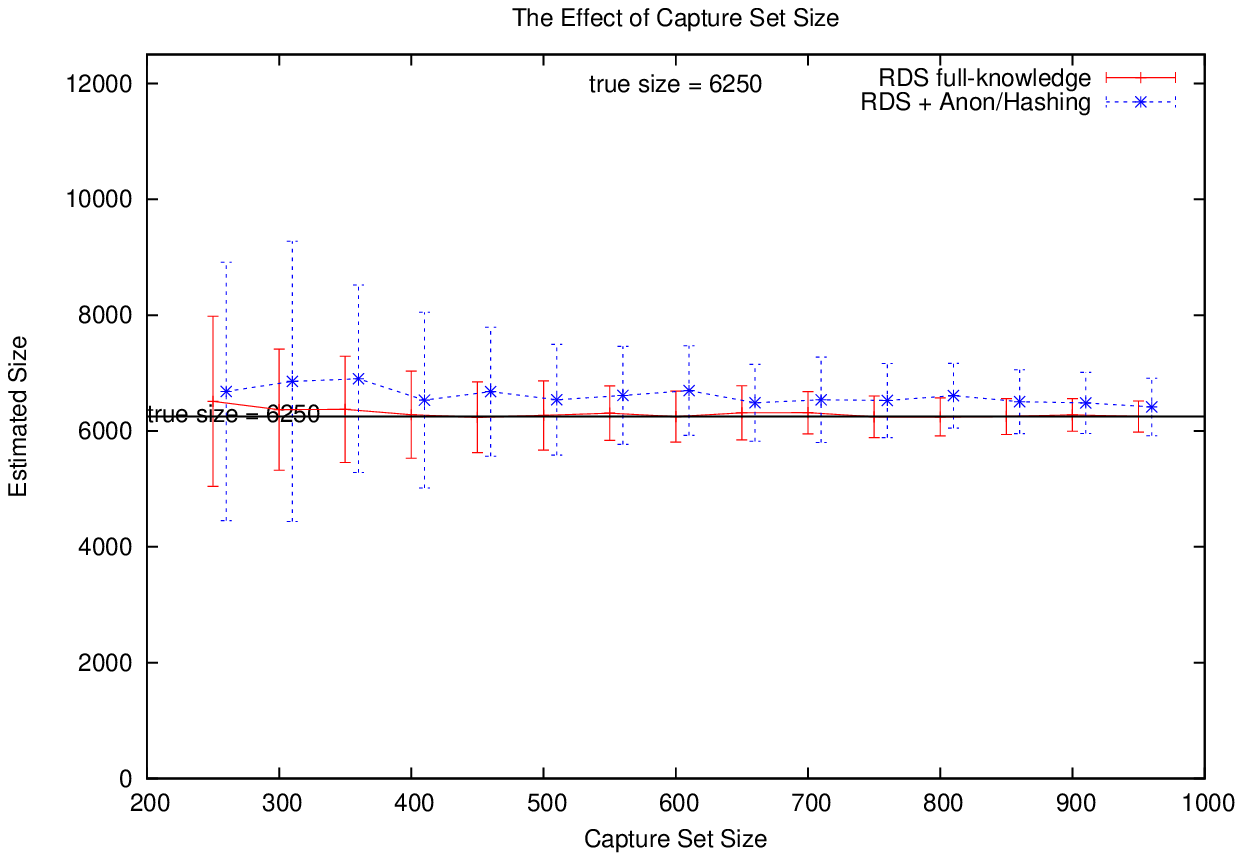}} &
\subcaptionbox{Actual population size 12{,}500\label{rds12}}{\includegraphics[width = 3.2in]{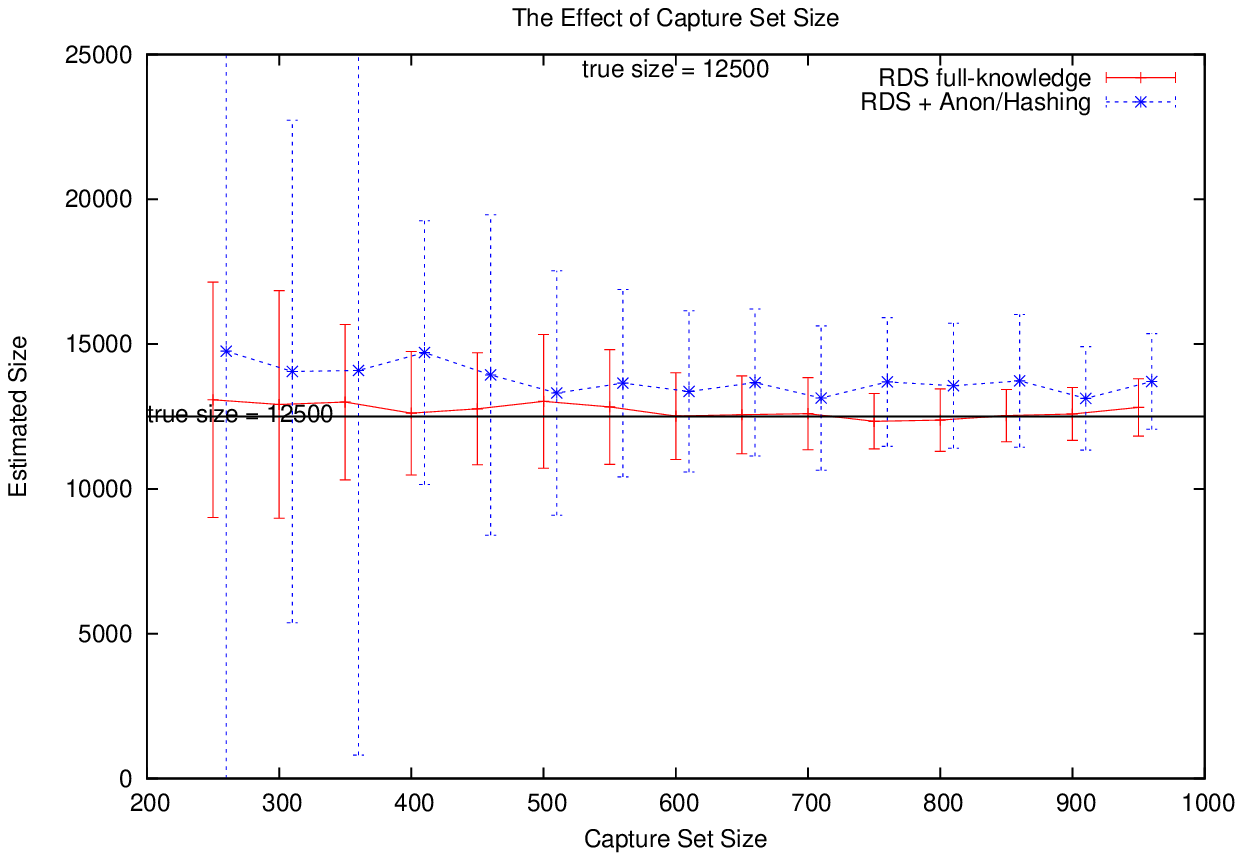}} \\
\subcaptionbox{Actual population size 25{,}000\label{rds25}}{\includegraphics[width = 3.2in]{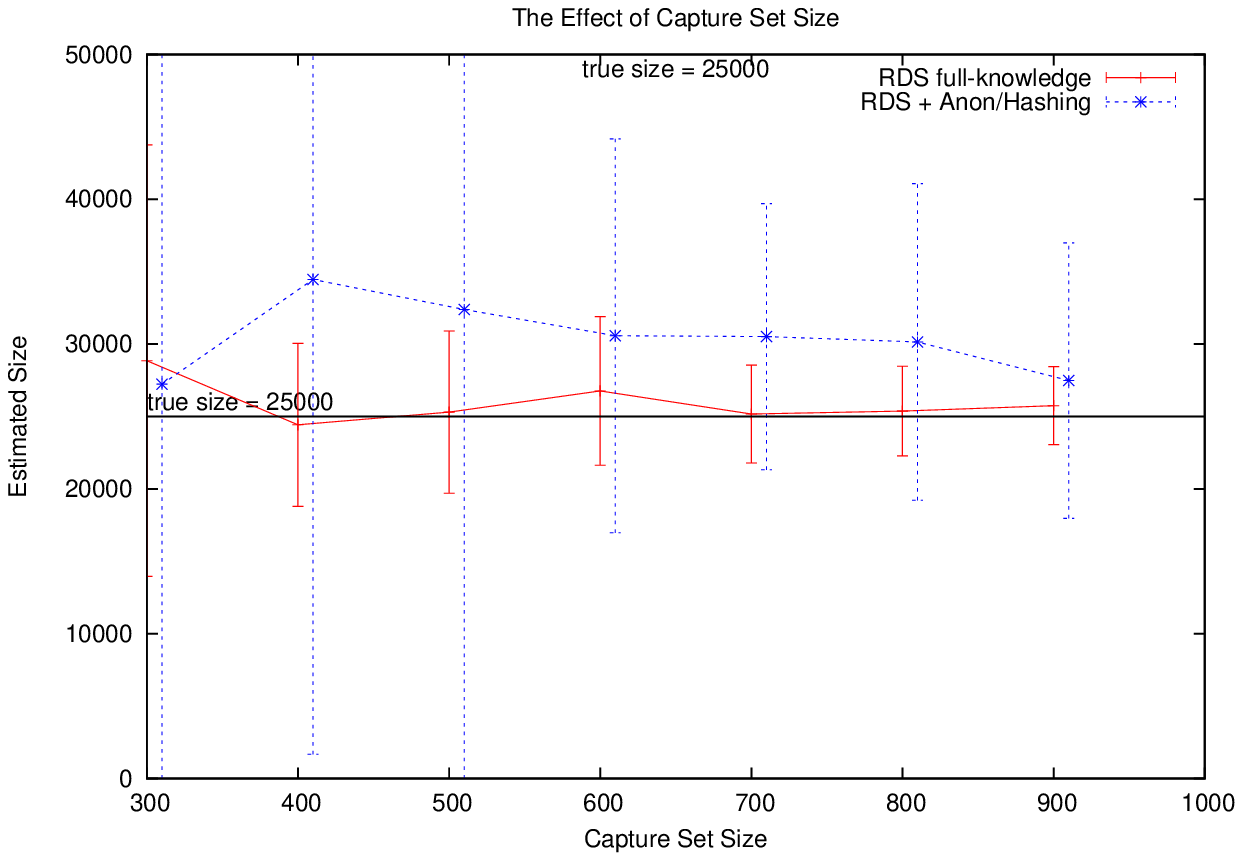}} &
\subcaptionbox{Actual population size 50{,}000\label{rds50}}{\includegraphics[width = 3.2in]{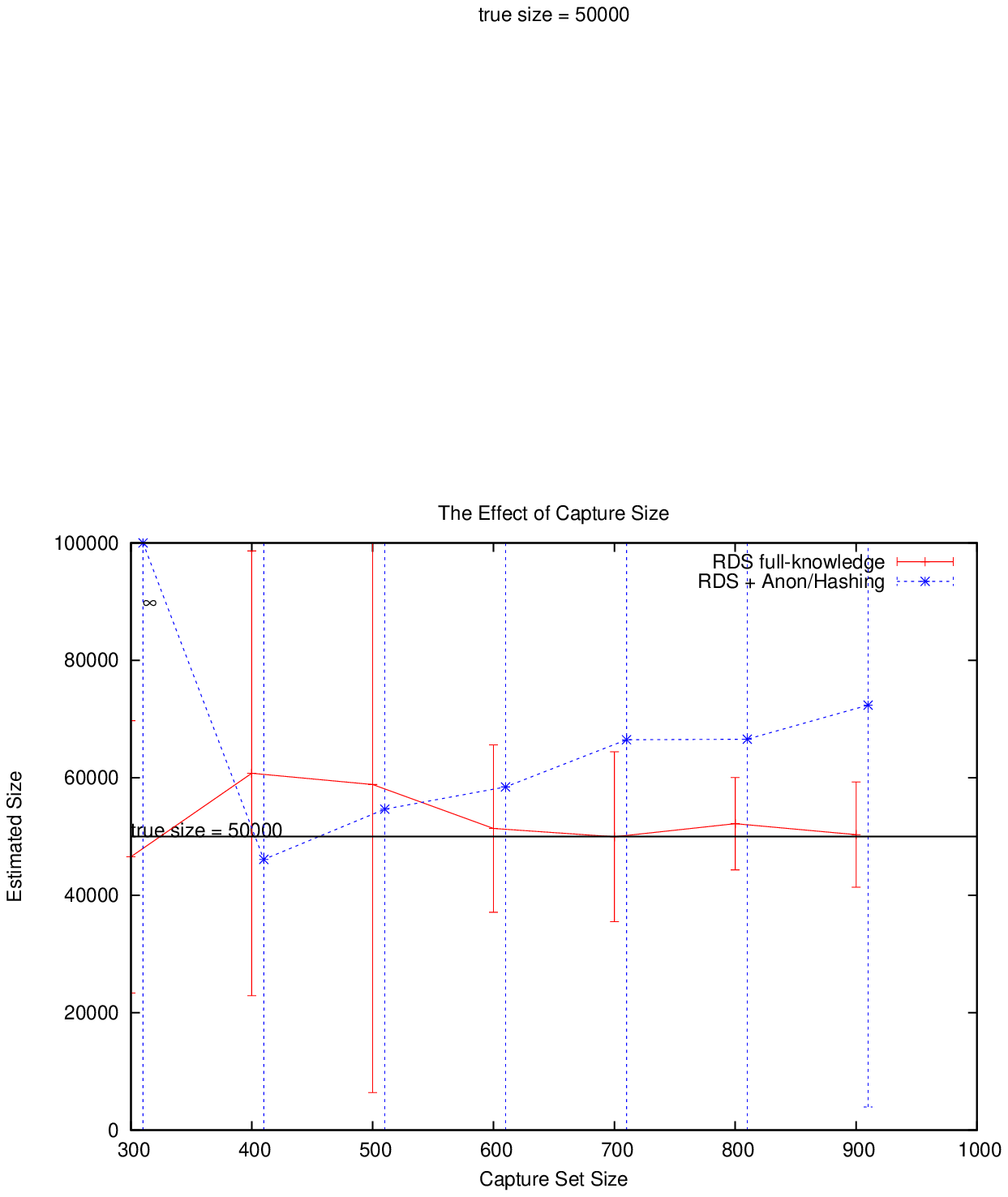}}
\end{tabular}
\caption{The impact of capture set size on population size estimate}
\label{rds-graphs}
\end{figure}

Figure \ref{seeds-graphs} shows how population estimates change as the {\bf number of seeds} varies from 
2 to 30 subjects inside networked populations
 comprised of (a) 6{,}250, (b) 12{,}500, (c) 25{,}000, and (d) 50{,}000 individuals.  In these experiments,
the size of the RDS capture sample $n_0=500$, the number of coupons given to each subject $c=3$,
the number of reports provided by each subject was $p=25$, and the hash space size $|H|=3{,}125$.
Notable in these results is that, for populations below 12{,}500, the number of seeds has little effect on the reliability or accuracy of the estimates. At or below 12{,}500, the estimation procedure produces relatively accurate predictions, with hashed results overestimating the true population size by 10-15 percent. In general, these results bode well for the use of this method under ordinary RDS implementation where 5-10 initial seeds are a regular rule of thumb. Above 12{,}500, the unhashed results remain relatively accurate, but the overall variance across simulation runs grows as the network population grows. The variation in the hashed results for these same large networks is itself very large, and seemingly unaffected by the number of seeds.  

\begin{figure}[p]
\centering
\begin{tabular}{cccc}
\subcaptionbox{Actual population size 6{,}250\label{seeds6}}{\includegraphics[width = 3.2in]{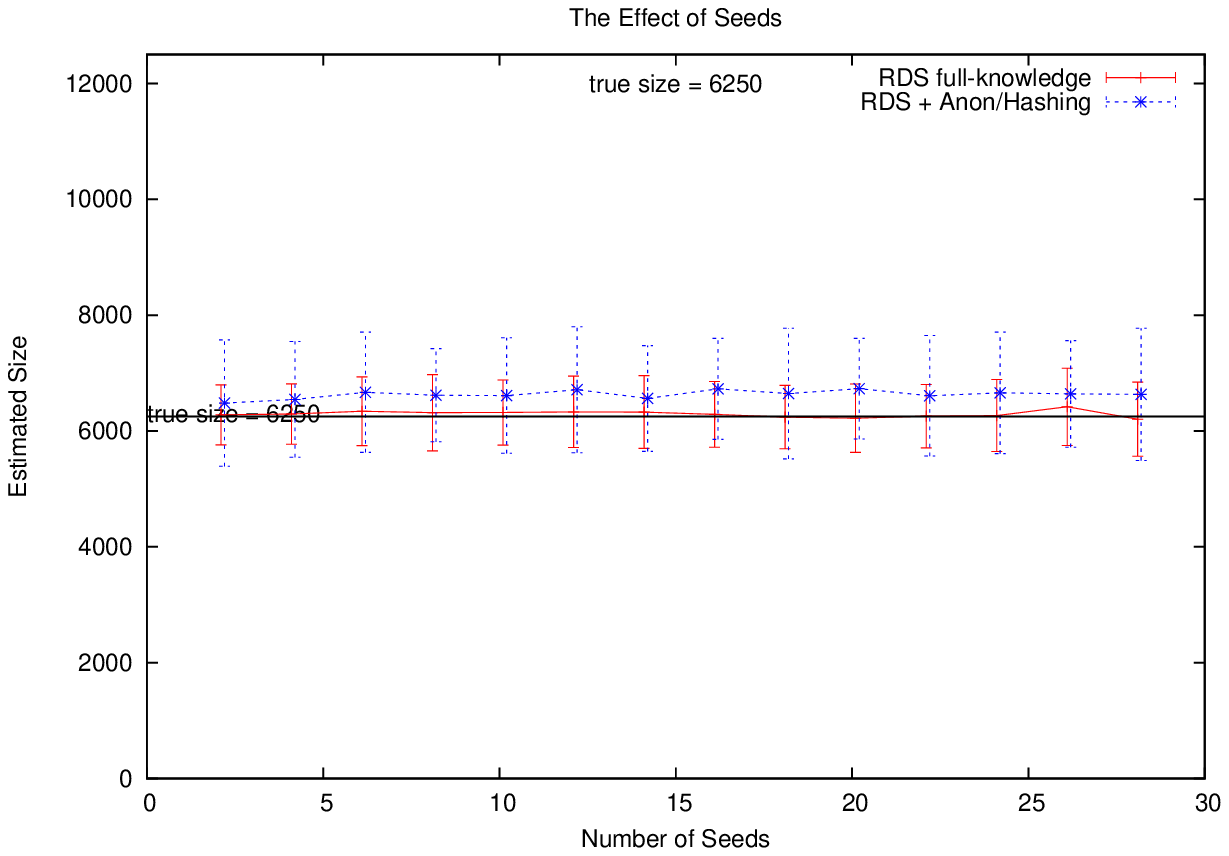}} &
\subcaptionbox{Actual population size 12{,}500\label{seeds12}}{\includegraphics[width = 3.2in]{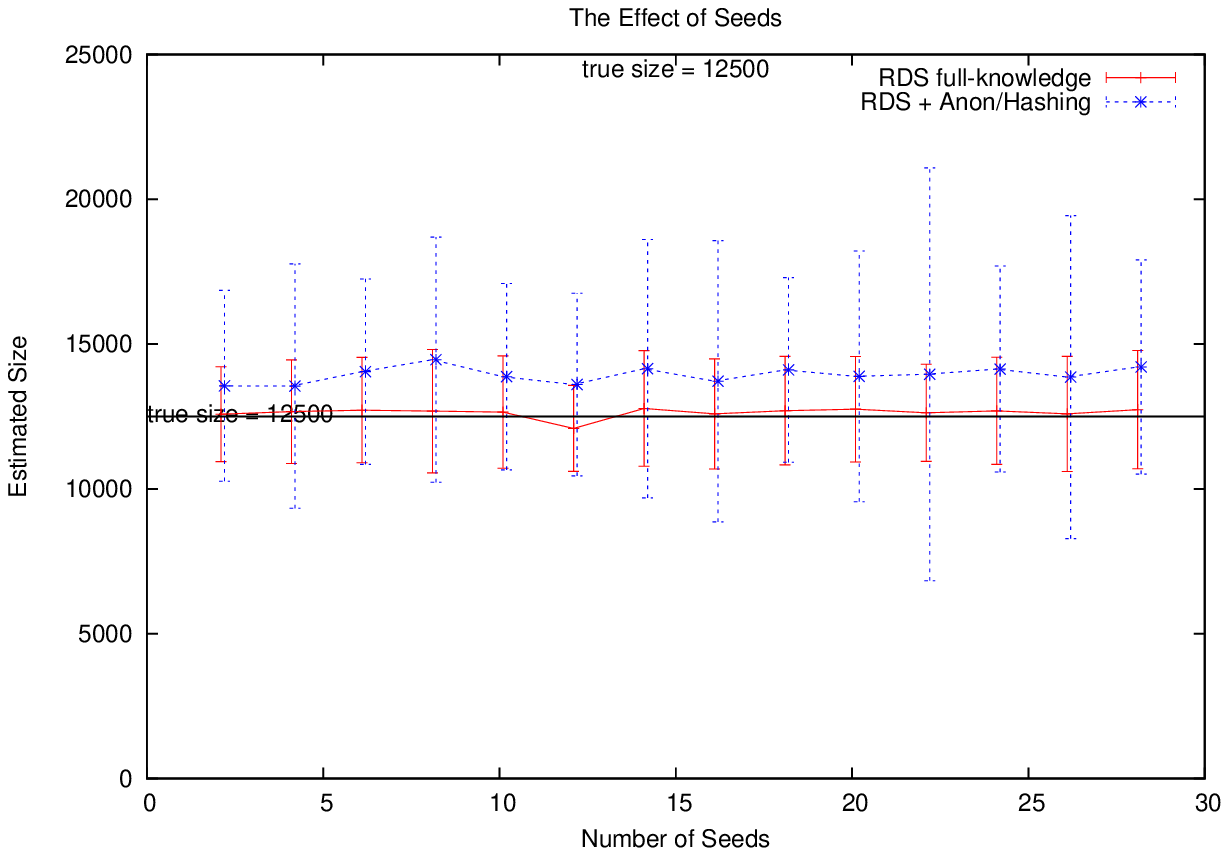}} \\
\subcaptionbox{Actual population size 25{,}000\label{seeds25}}{\includegraphics[width = 3.2in]{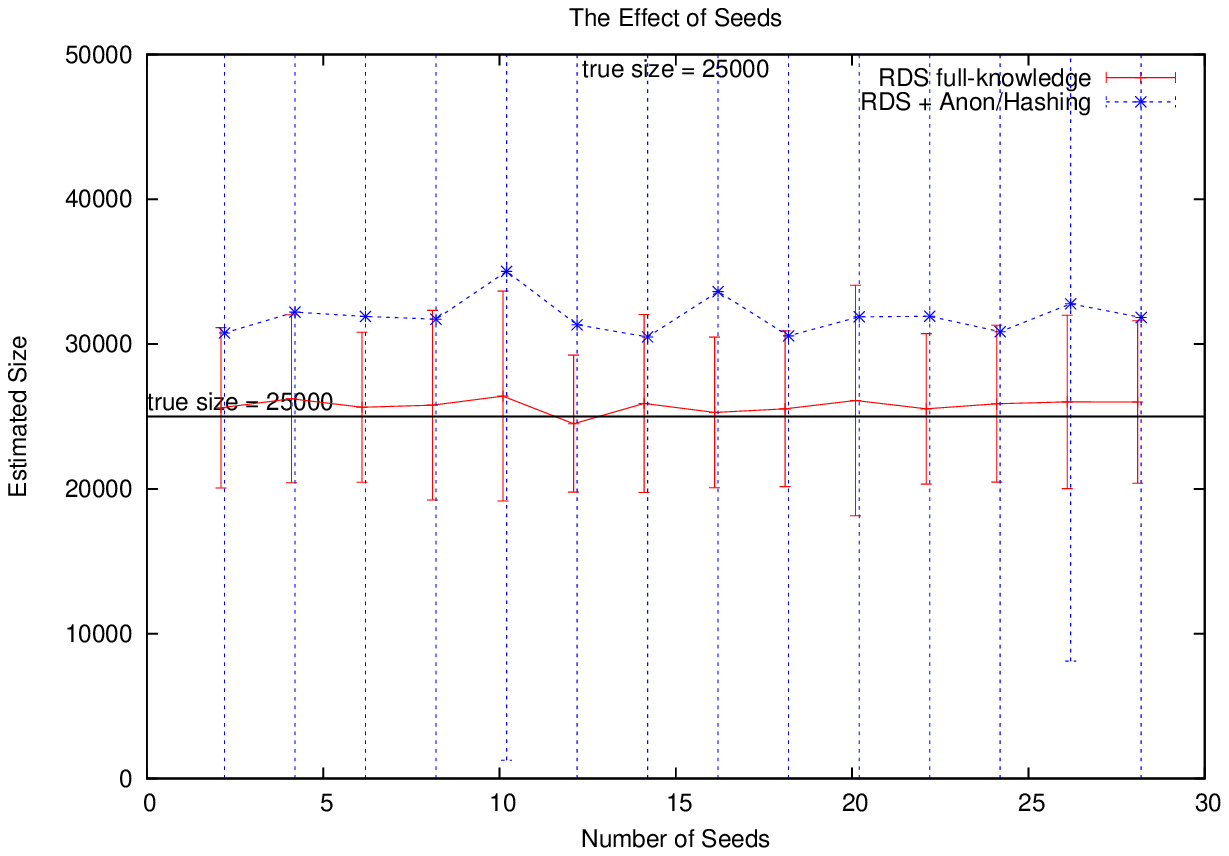}} &
\subcaptionbox{Actual population size 50{,}000\label{seeds50}}{\includegraphics[width = 3.2in]{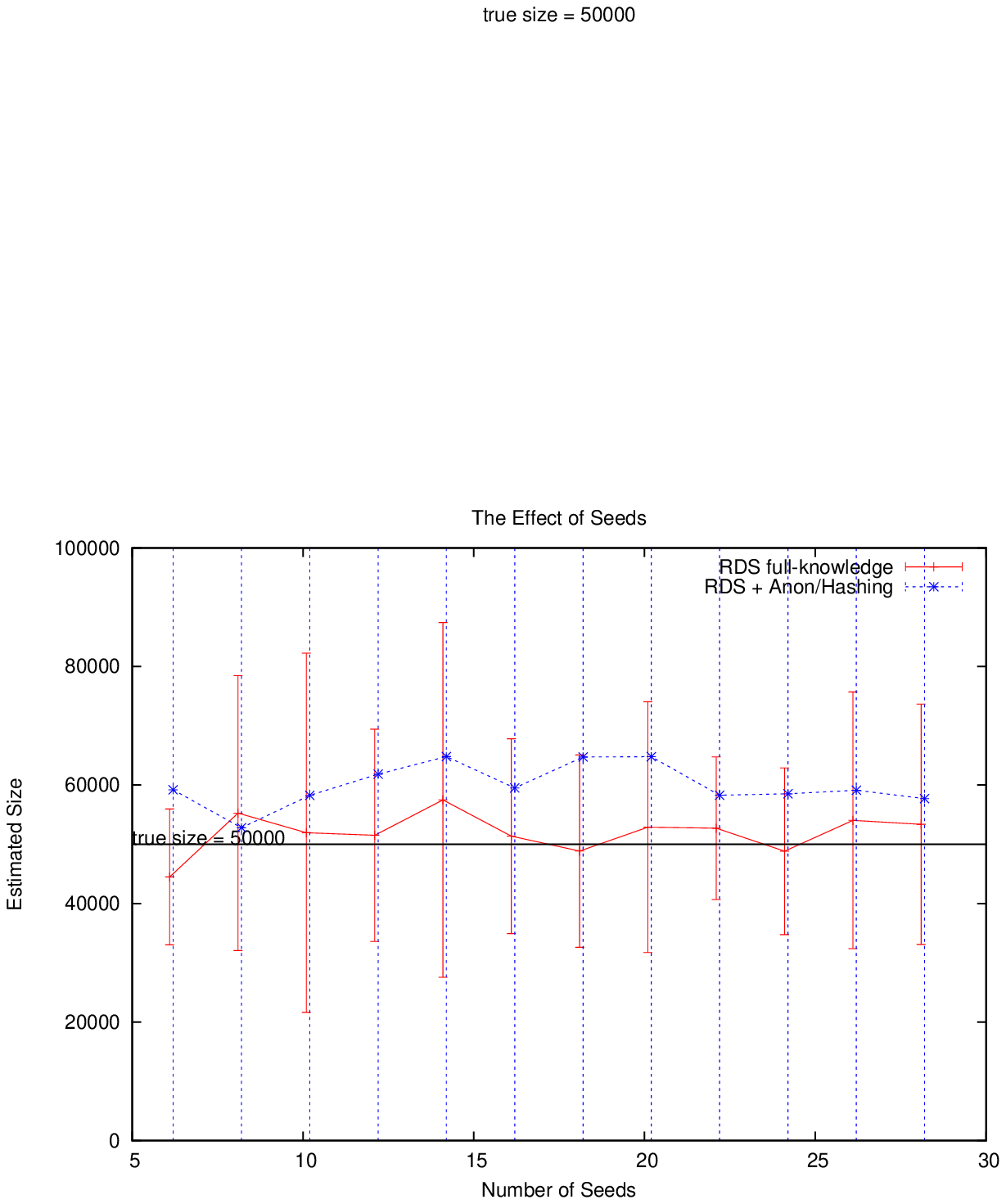}}
\end{tabular}
\caption{The impact of number of initial seeds on population size estimate}
\label{seeds-graphs}
\end{figure}

Figure \ref{coupons-graphs} shows how population estimates change as the {\bf number of coupons} given to each subject is varied from 
1 to 5, inside networked populations
 comprised of (a) 6{,}250, (b) 12{,}500, (c) 25{,}000, and (d) 50{,}000 individuals.  In these experiments,
the size of the RDS capture sample $n_0=500$, the number of seeds $s=6$, 
the number of reports provided by each subject is $p=25$, and the hash space size $|H|=3{,}125$.
Here, as with the number of initial seeds, the method performs well for network sizes at or below 12{,}500, and the number of coupons has little predictable effect on the reliability or accuracy of the estimates at this threshold. Here too hashed results overestimate the actual network size by the same 10-15 percent, but produce tolerable variance levels independent of the number of coupons given (below 5). Above 12{,}500, the unhashed results remain relatively accurate, and are robust against changes in the number of referrals allowed to each participant. The variances in the hashed results for populations above 12{,}500 are very large, and this concern is unaffected by the number of seeds.  

\begin{figure}[p]
\centering
\begin{tabular}{cccc}
\subcaptionbox{Actual population size 6{,}250\label{coupons6}}{\includegraphics[width = 3.2in]{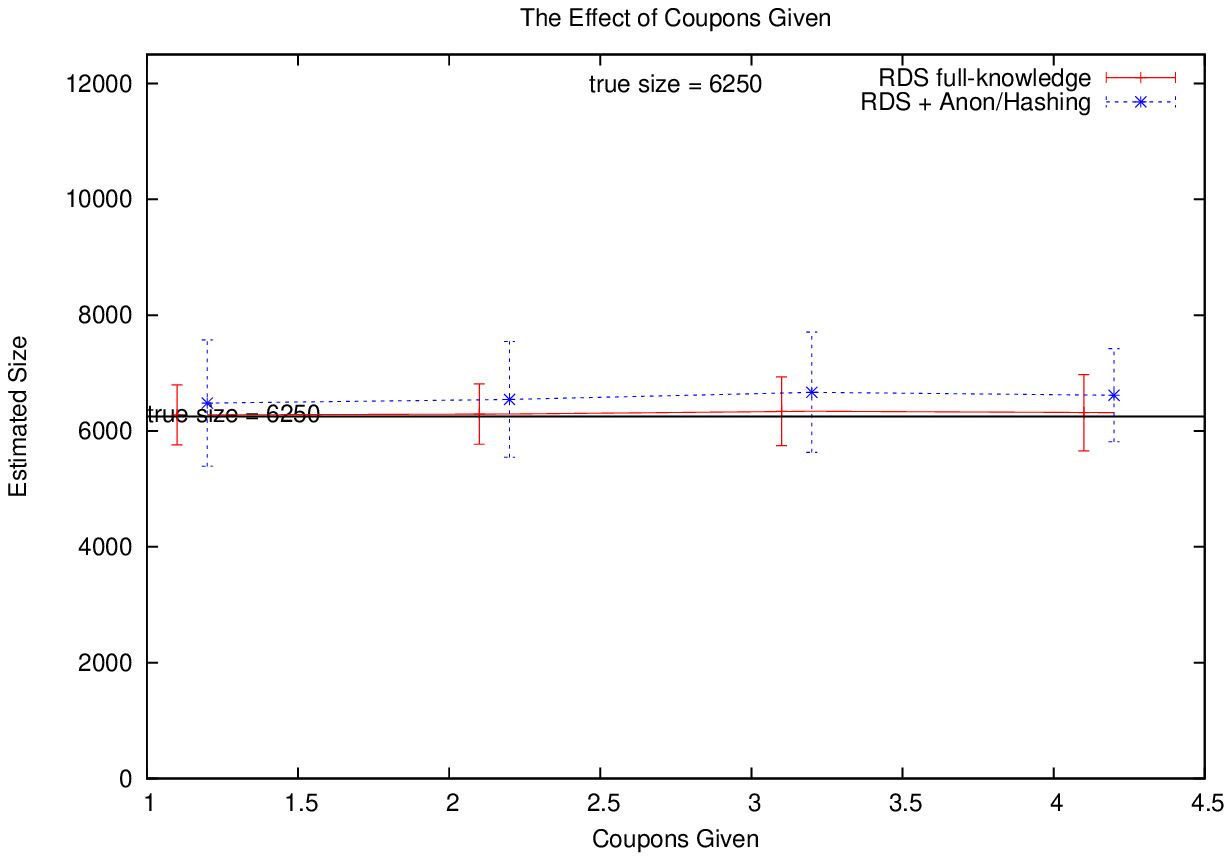}} &
\subcaptionbox{Actual population size 12{,}500\label{coupons12}}{\includegraphics[width = 3.2in]{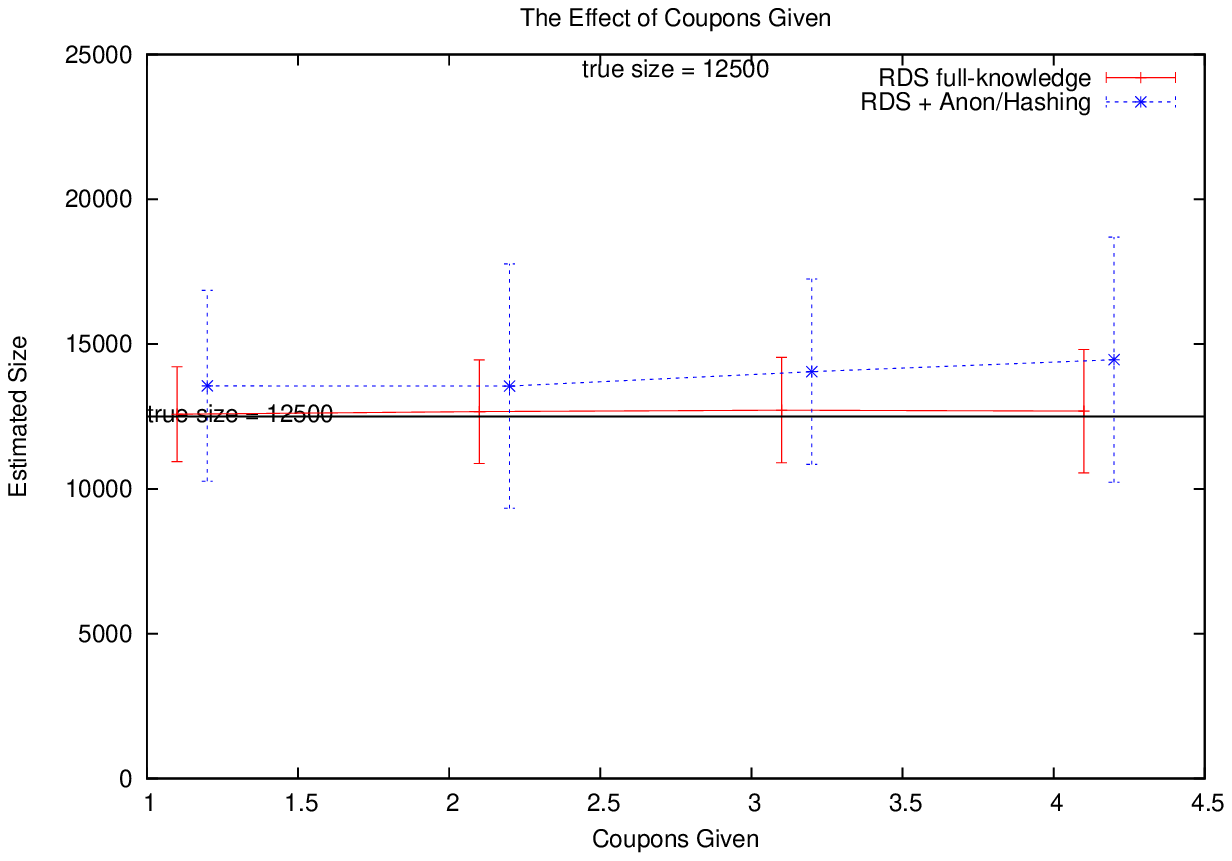}} \\
\subcaptionbox{Actual population size 25{,}000\label{coupons25}}{\includegraphics[width = 3.2in]{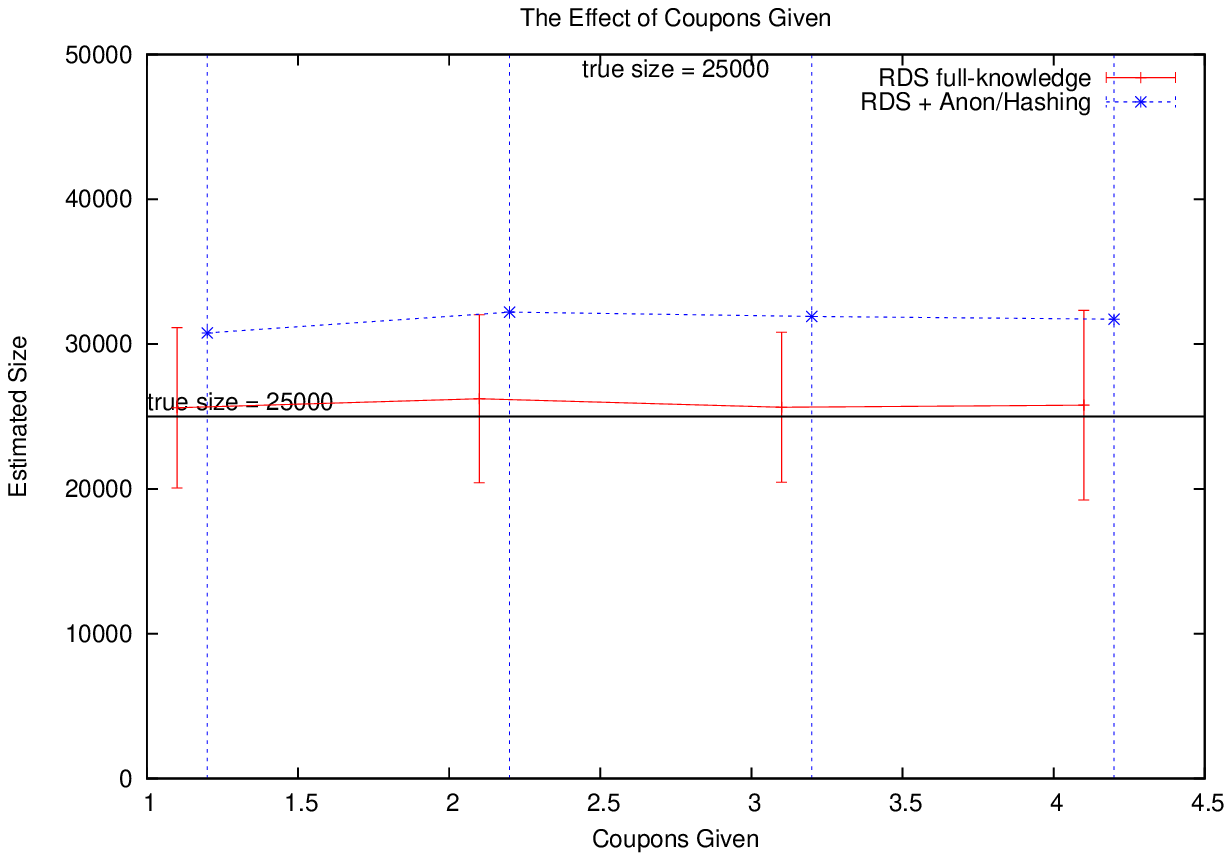}} &
\subcaptionbox{Actual population size 50{,}000\label{coupons50}}{\includegraphics[width = 3.2in]{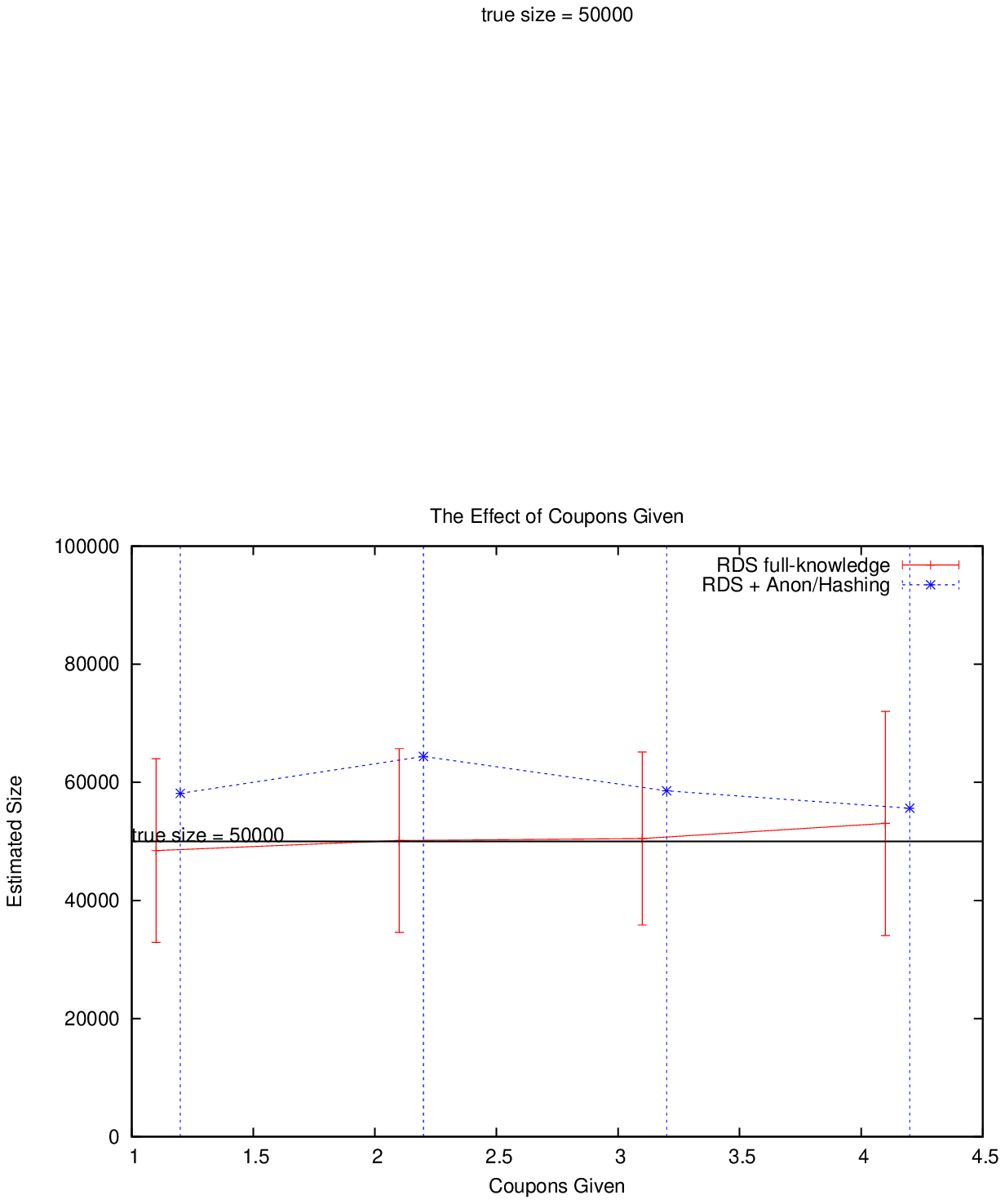}}
\end{tabular}
\caption{The impact of number of coupons per subject on population size estimate}
\label{coupons-graphs}
\end{figure}

Figure \ref{reports-graphs} shows how population estimates change as the {\bf number of reports} requested from each subject is varied from 1 to 30, inside  networked populations comprised of (a) 6{,}250, (b) 12{,}500, (c) 25{,}000, and (d) 50{,}000 individuals.  In these experiments, the size of the RDS capture sample $n_0=500$, the number of seeds $s=6$, the number of coupons given to each subject $c=3$, and the hash space size $|H|=3{,}125$. These results show a similar pattern to those discovered when varying the number of initial seeds or the number of coupon referrals allowed to each participant. For network sizes below 12{,}500, both unhashed and hashed results perform well, with variance growing as network size is increased regardless of the number of reports requested from each respondent. We note that for unhashed results, variance diminishes as the number of reports grows to 5, but that little is gained in reliability or accuracy of the estimates by allowing for more than 5 reports. Above 12{,}500, variance for the hashed results increases considerably, though unhashed results remain more robust. This is true in the extreme at the highest level of population tested (e.g. networks of size 50{,}000). 

\begin{figure}[p]
\centering
\begin{tabular}{cccc}
\subcaptionbox{Actual population size 6{,}250\label{reports6}}{\includegraphics[width = 3.2in]{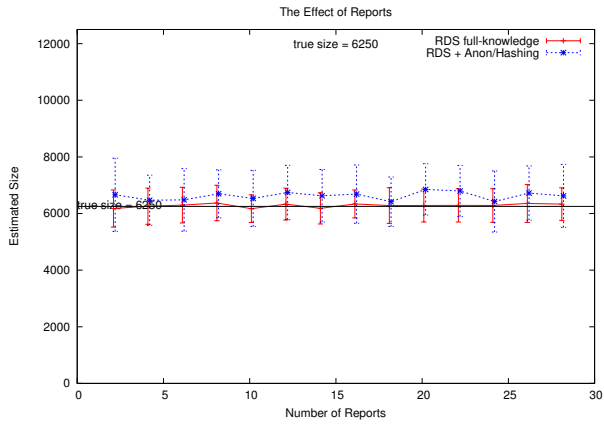}} &
\subcaptionbox{Actual population size 12{,}500\label{reports12}}{\includegraphics[width = 3.2in]{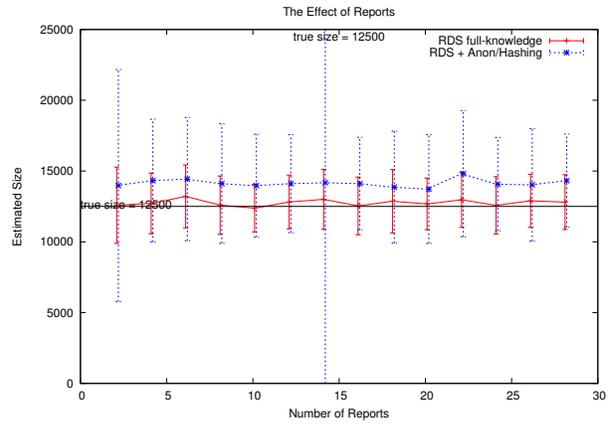}} \\
\subcaptionbox{Actual population size 25{,}000\label{reports25}}{\includegraphics[width = 3.2in]{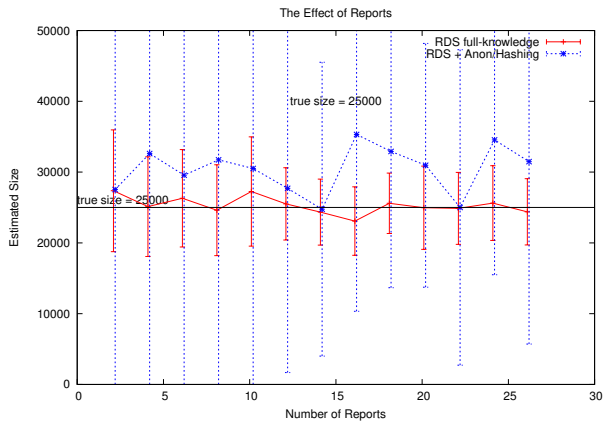}} &
\subcaptionbox{Actual population size 50{,}000\label{reports50}}{\includegraphics[width = 3.2in]{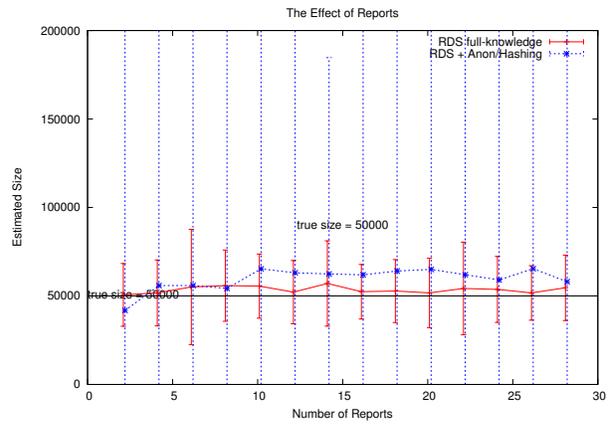}}
\end{tabular}
\caption{The impact of number of reports per subject on population size estimate}
\label{reports-graphs}
\end{figure}

Figure \ref{hash-graphs} shows how population estimates change as the {\bf hash space size} is varied from 100 to 4000, inside  networked populations comprised of (a) 6{,}250, (b) 12{,}500, (c) 25{,}000, and (d) 50{,}000 individuals.  In these experiments, the size of the RDS capture sample $n_0=500$, the number of seeds $s=6$, the number of coupons given to each subject $c=3$, and the number of reports provided by each subject is $p=25$. These results show a similar pattern to those discovered when varying the number of initial seeds or the number of coupon referrals allowed to each participant. Little is gained by increasing the size of the hashing space up to 4{,}000. However, the unhashed results perform well in terms of both accuracy and variance, suggesting that further increases in the overall hashing space ought to help narrow the variance of the hashed results and improve the accuracy of the estimate as hashing space is increased.

\begin{figure}[h]
\centering
\begin{tabular}{cccc}
\subcaptionbox{Actual population size 6{,}250\label{hash6}}{\includegraphics[width = 3.2in]{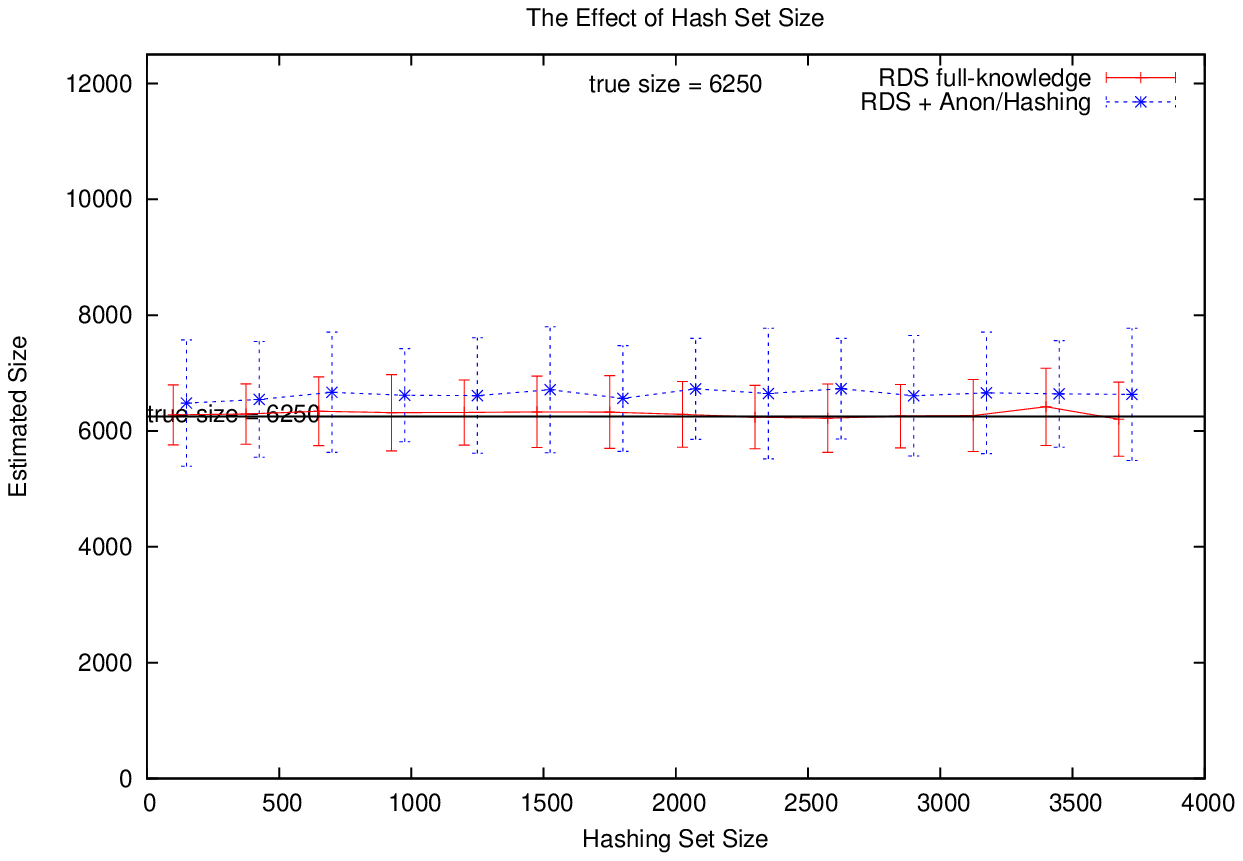}} &
\subcaptionbox{Actual population size 12{,}500\label{hash12}}{\includegraphics[width = 3.2in]{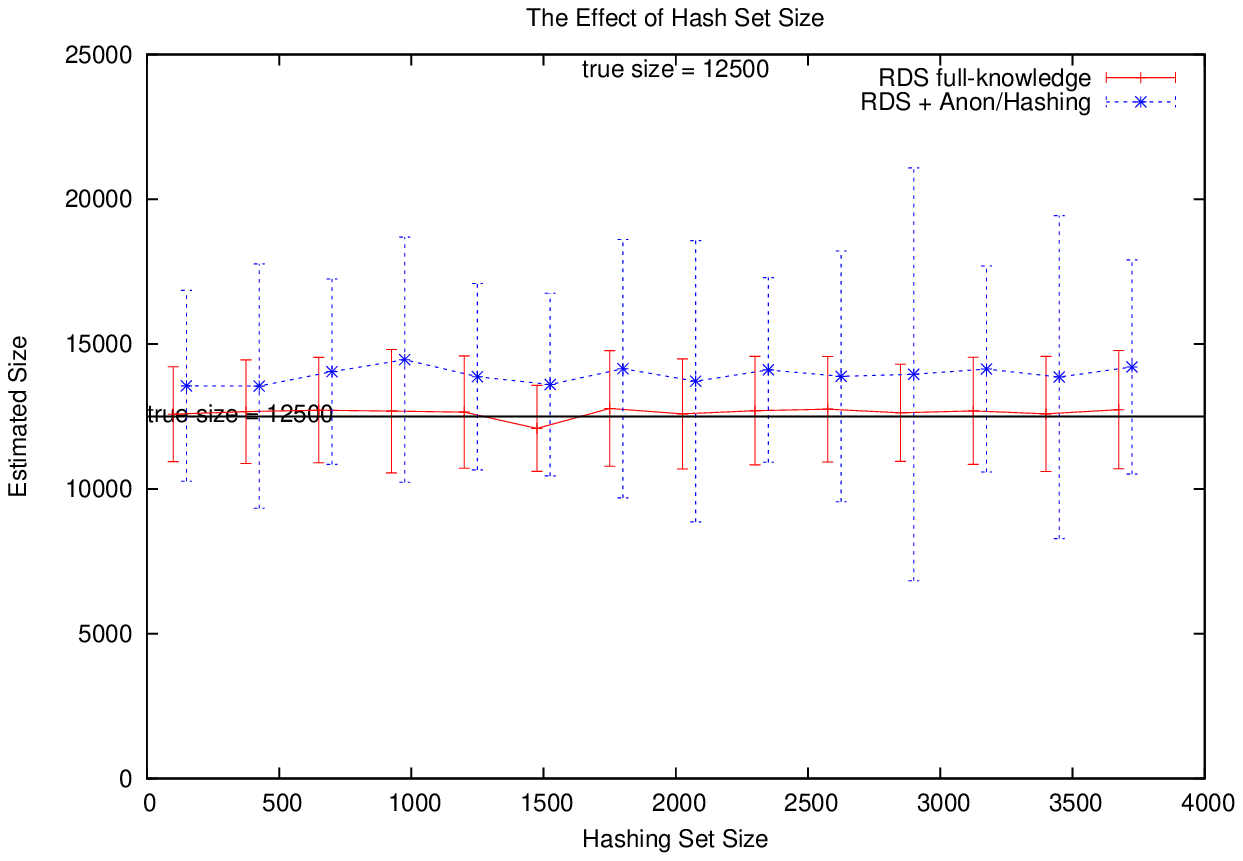}} \\
\subcaptionbox{Actual population size 25{,}000\label{hash25}}{\includegraphics[width = 3.2in]{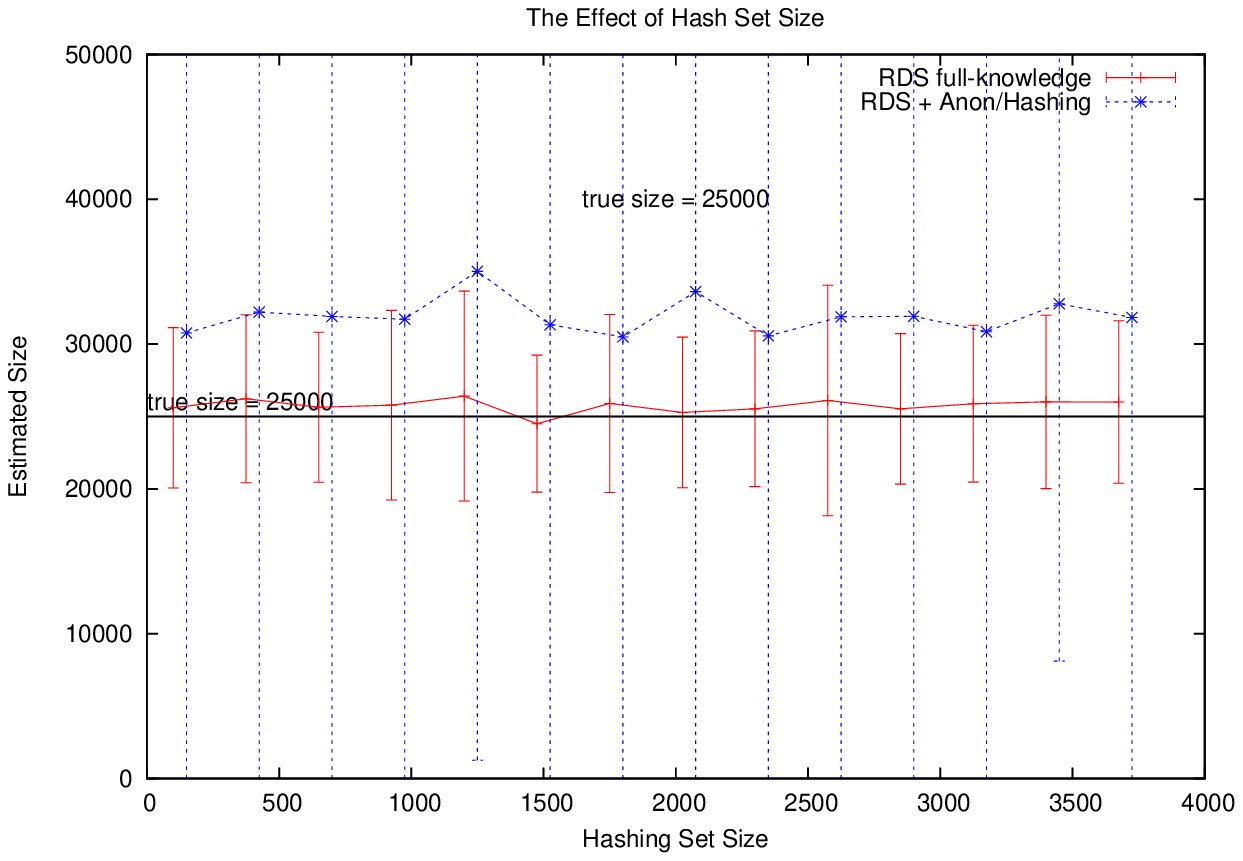}} &
\subcaptionbox{Actual population size 50{,}000\label{hash50}}{\includegraphics[width = 3.2in]{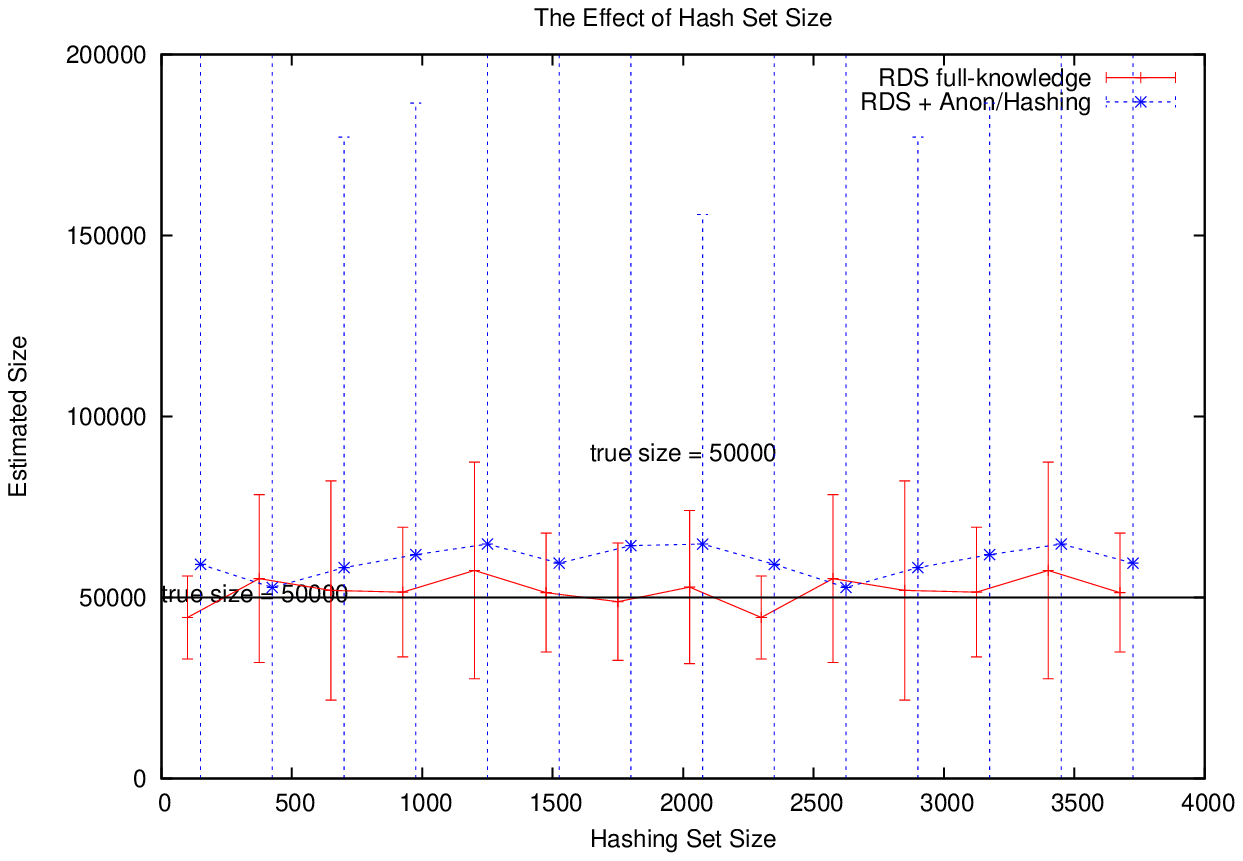}}
\end{tabular}
\caption{The impact of hash set size on population size estimate}
\label{hash-graphs}
\end{figure}

\subsection{Software Platform}

The software was written in Java using the JUNG \cite{o2003jung} library.  The source code for the programs used in this research is available on 
Github, by permission of the authors.

\section{Discussion and Limitations}

In a number of ways, these results are encouraging of the ability of this method to accurately and reliably estimate network size from ordinary RDS samples, despite pointing to the need for additional experiments. Under conditions commonly associated with RDS implementation in research on hidden populations---the standard implementation protocols of 5 to 10 seeds, 3 to 5 coupons, and sample sizes of 500 or more---the method proposed here resulted in relative robust and mainly reliable estimates for hidden populations of size 12{,}500 or smaller, even under conditions of hashing/anonymity. The results shown in Figure \ref{rds-graphs} suggest that increased RDS sample size could potentially raise the limiting population threshold to 25{,}000 or higher without varying the number of seeds or allowed referrals. What's more, the performance of the unhashed estimator in larger populations suggests that a greater hashing space (and thus lower number of expected false matches) could also raise this population threshold even further. We note that a very large hashing space is provided by the increased use of mobile phones among the populations of interest. As described in an earlier paper \cite{TELEFUNKEN2012}, the encoding of phone numbers allows for a very large hashing space that is both one way and anonymous, and which can, to an extent, be assumed to be randomly distributed in the population. The use of a significant number of digits, say seven, in an even/odd, 0-4/5-9 encoding would provide a hashing space of $2^{14}$, which could potentially bring the hashed results for larger networks more in line with the unhashed results seen at these same sizes.

We also show that the method can avoid the potential problem of negative population estimates  without resorting to additional sampling, simply by implementing a bootstrap re-sampling of the existing project data. Using this technique, different fractions of the RDS-based capture set $\uppsi S$ and recapture multiset $\uppsi rS$ are obtained without conducting additional data collection. In effect, this procedure can be thought of as multiple independent trials carried out throughout the process. With each boostrap pair of $(S,rS)$, we have, in effect, a distinct population estimate. This allows us to use the mulitple trees common in RDS research to minimize the influence of negative estimations arriving from one or another of these trees, or unusual interaction of reports and recruitment between trees. Not shown here is the fact that large standard deviations in the bootstrapping results could potentially serve as a red flag for both the overall population size estimate and the RDS sample as a whole. 

Several limitations on the above work must be noted, however. Primarily, this paper examines only one family of graphs that are assumed to connect the population of interest (with the larger ambient population). While there is reason to believe that fat-tailed degree distributions are prevalent in hidden populations, further experimentation is needed to ascertain the performance of the estimation procedure described here for other families of graphs. It is also notable that the use of 100 runs at each experiment setting allowed us to calculate the standard deviation and overall variance associated with a particular estimate, but such measures are of limited value when most implementations are for a single data episode.  Under single sample conditions, it is important to acknowledge that choosing a single run from the many trials actually invokes greater stochasticity than seen in the standard deviation bars in the Figures presented above.

And finally, we did not examine the effects of clustering in any of the experiments \cite{CLUSTERING1971}. Clustering has been as a crucial consideration in drawing RDS samples, and is likely to have an effect on the estimation results shown here. 

\section{Conclusions and  Future Work}

This paper extends and formalizes previous work on a one-step, network-based population estimation procedure that can be employed under conditions of anonymity \cite{TELEFUNKEN2012}. It employs a modified capture-recapture methodology for estimating the size of a hidden population from an RDS sample, allowing it to be employed in situations where sub-populations of interest to health officials remain hidden within a larger ambient populations. Where successful, it provides reasonable population estimates that can be used in conjunction with RDS-produced prevalence estimates to give public health workers information necessary for health promotion. The easy integration of the method with current data collection strategies for key populations and other hidden groups allows for greater cost-effectiveness for size estimation procedures with minimal additional participant burden. Prior work by our own team \cite{curtis_commercial_2008, wendel_dynamics_2011} and others \cite{sulaberidze_population_2016} has shown that this technique can be employed under ordinary fieldwork conditions common to hidden populations, and the simulation experiments here show that reasonably accurate population estimates can be derived from relatively small sample sizes (between 2 and 8 percent of the hidden population total) under conditions of subject anonymity. Population size estimation for hidden and hard-to-reach populations is of continuing importance to health officials because health problems are often found to be concentrated in such populations, while their “hiddenness” necessarily frustrates health outreach efforts \cite{magnani_review_2005, merli_sampling_2016}. As such, this technique should be of interest to health researchers and policy makers working with such hard to reach groups.

Given both the results and the limitations discussed above, three avenues of future work will be pursued. In particular, future experiments will explore 1) other random graph families (e.g. Erd\H{o}s-R\'{e}nyi, Exponential Random Graphs and others normally associated with human social networks); 2) the effect of clustering within these graph families as it affects both hashed and unhashed estimates; and 3) the impact of greater hashing space that may allow for better results (i.e. lower variance) as network size increases. These steps will help clarify the accurate use of these techniques, and perhaps extend their use to larger population sizes. 

\textbf{Acknowledgments:} This work was supported by the NIH National Institute on Drug Abuse of the National Institutes of Health (grant number R01DA037117) and the NIH National Institute for General Medical Sciences (grant number R01GM118427). The content is solely the responsibility of the authors and does not necessarily represent the official views of the National Institutes of Health. Special thanks to Travis Wendel, Ric Curtis, and Patrick Habecker for their work on earlier versions of this methodology.

\clearpage
\medskip
 
\bibliographystyle{unsrt}
\bibliography{paper}

\begin{thebibliography}{10}

\bibitem{magnani_review_2005}
R.~Magnani, K.~Sabin, T.~Saidel, and D.~Heckathorn.
\newblock Review of sampling hard-to-reach and hidden populations for {HIV}
  surveillance.
\newblock {\em Aids}, 19:S67, 2005.

\bibitem{dombrowski_topological_2013}
Kirk Dombrowski.
\newblock Topological and {Historical} {Considerations} for {Infectious}
  {Disease} {Transmission} among {Injecting} {Drug} {Users} in {Bushwick},
  {Brooklyn} ({USA}).
\newblock {\em World Journal of AIDS}, 03(01):1--9, 2013.

\bibitem{reluga_reservoir_2007}
T.~Reluga, R.~Meza, D.b. Walton, and A.p. Galvani.
\newblock Reservoir interactions and disease emergence.
\newblock {\em Theoretical Population Biology}, 72(3):400--408, November 2007.

\bibitem{bonin_typology_2009}
Jean-Pierre Bonin, Louise Fournier, and Regis Blais.
\newblock A {Typology} of {Mentally} {Disordered} {Users} of {Resources} for
  {Homeless} {People}: {Towards} {Better} {Planning} of {Mental} {Health}
  {Services}.
\newblock {\em Administration and Policy in Mental Health and Mental Health
  Services Research}, 36(4):223--235, July 2009.
\newblock 00009 WOS:000266918400001.

\bibitem{burt_critical_1995}
Martha~R. Burt.
\newblock Critical {Factors} in {Counting} the {Homeless}.
\newblock {\em American Journal of Orthopsychiatry}, 65(3):334--339, 1995.

\bibitem{potterat_aids_1993}
J~J Potterat, D~E Woodhouse, R~B Rothenberg, S~Q Muth, W~W Darrow, J~B Muth,
  and J~U Reynolds.
\newblock {AIDS} in {Colorado} {Springs}: is there an epidemic?
\newblock {\em AIDS (London, England)}, 7(11):1517--1521, November 1993.

\bibitem{abdul2014estimating}
Abu~S Abdul-Quader, Andrew~L Baughman, and Wolfgang Hladik.
\newblock Estimating the size of key populations: current status and future
  possibilities.
\newblock {\em Current Opinion in HIV and AIDS}, 9(2):107--114, 2014.

\bibitem{law_spatial_2004}
D.~C.~G. Law, M.~L. Serre, G.~Christakos, P.~A. Leone, and W.~C. Miller.
\newblock Spatial analysis and mapping of sexually transmitted diseases to
  optimise intervention and prevention strategies.
\newblock {\em Sexually Transmitted Infections}, 80(4):294--299, August 2004.

\bibitem{zohrabyan_determinants_2013}
Lev Zohrabyan, Lisa~Grazina Johnston, Otilia Scutelniciuc, Alexandrina Iovita,
  Lilia Todirascu, Tatiana Costin, Valeriu Plesca, Tatiana Cotelnic-Harea, and
  Gabriela Ionascu.
\newblock Determinants of {HIV} {Infection} {Among} {Female} {Sex} {Workers} in
  {Two} {Cities} in the {Republic} of {Moldova}: {The} {Role} of {Injection}
  {Drug} {Use} and {Sexual} {Risk}.
\newblock {\em AIDS and behavior}, March 2013.

\bibitem{darke_self-report_1998}
Shane Darke.
\newblock Self-report among injecting drug users: {A} review.
\newblock {\em Drug and Alcohol Dependence}, 51(3):253--263, August 1998.

\bibitem{harwood_sampling_2012}
Eileen~M. Harwood, Keith~J. Horvath, Cari Courtenay-Quirk, Holly Fisher, Rachel
  Kachur, Mary McFarlane, Bryn Meyer, BR~Simon Rosser, and Ann O’Leary.
\newblock Sampling hidden populations: lessons learned from a telephone-based
  study of persons recently diagnosed with {HIV} ({PRDH}).
\newblock {\em International Journal of Social Research Methodology},
  15(1):31--40, 2012.

\bibitem{larson_indirect_1994}
Ann Larson, Adele Stevens, and Grant Wardlaw.
\newblock Indirect estimates of ‘hidden’ populations: {Capture}-recapture
  methods to estimate the numbers of heroin users in the {Australian} capital
  territory.
\newblock {\em Social Science \& Medicine}, 39(6):823--831, September 1994.

\bibitem{vuylsteke_capturerecapture_2010}
B.~Vuylsteke, H.~Vandenhoudt, L.~Langat, G.~Semde, J.~Menten, F.~Odongo,
  A.~Anapapa, L.~Sika, A.~Buve, and M.~Laga.
\newblock Capture–recapture for estimating the size of the female sex worker
  population in three cities in {Côte} d’{Ivoire} and in {Kisumu}, western
  {Kenya}.
\newblock {\em Tropical Medicine \& International Health}, 15(12):1537--1543,
  2010.

\bibitem{biernacki_snowball_1981}
Patrick Biernacki and Dan Waldorf.
\newblock Snowball {Sampling}: {Problems} and {Techniques} of {Chain}
  {Referral} {Sampling}.
\newblock {\em Sociological Methods \& Research}, 10(2):141 --163, November
  1981.

\bibitem{platt_methods_2006}
Lucy Platt, Martin Wall, Tim Rhodes, Ali Judd, Matthew Hickman, Lisa~G
  Johnston, Adrian Renton, Natalia Bobrova, and Anya Sarang.
\newblock Methods to recruit hard-to-reach groups: comparing two chain referral
  sampling methods of recruiting injecting drug users across nine studies in
  {Russia} and {Estonia}.
\newblock {\em Journal of Urban Health: Bulletin of the New York Academy of
  Medicine}, 83(6 Suppl):i39--53, November 2006.

\bibitem{haley_venue-based_2014}
Danielle~F. Haley, Carol Golin, Wafaa El-Sadr, James~P. Hughes, Jing Wang,
  Malika Roman~Isler, Sharon Mannheimer, Irene Kuo, Jonathan Lucas, Elizabeth
  DiNenno, and {others}.
\newblock Venue-based recruitment of women at elevated risk for {HIV}: an {HIV}
  prevention trials network study.
\newblock {\em Journal of Women's Health}, 23(6):541--551, 2014.

\bibitem{muhib_venue-based_2001}
Farzana~B. Muhib, Lillian~S. Lin, Ann Stueve, Robin~L. Miller, Wesley~L. Ford,
  Wayne~D. Johnson, Philip~J. Smith, Community Intervention Trial for
  Youth~Study Team, and {others}.
\newblock A venue-based method for sampling hard-to-reach populations.
\newblock {\em Public health reports}, 116(Suppl 1):216, 2001.

\bibitem{burnham_mortality_2006}
Gilbert Burnham, Riyadh Lafta, Shannon Doocy, and Les Roberts.
\newblock Mortality after the 2003 invasion of {Iraq}: a cross-sectional
  cluster sample survey.
\newblock {\em The Lancet}, 368(9545):1421--1428, 2006.

\bibitem{heckathorn_extensions_2007}
D.~D Heckathorn.
\newblock Extensions of {Respondent}-{Driven} {Sampling}: {Analyzing}
  {Continuous} {Variables} and {Controlling} for {Differential} {Recruitment}.
\newblock {\em Sociological Methodology}, 37(1):151--207, 2007.

\bibitem{RDS2002}
Douglas~D. Heckathorn.
\newblock Respondent-driven sampling ii: Deriving valid population estimates
  from chain-referral samples of hidden populations.
\newblock {\em Social Problems}, 2002.

\bibitem{salganik_sampling_2004}
Matthew~J. Salganik and Douglas~D. Heckathorn.
\newblock Sampling and {Estimation} in {Hidden} {Populations} {Using}
  {Respondent}-{Driven} {Sampling}.
\newblock {\em Sociological Methodology}, 34(1):193--239, December 2004.

\bibitem{abdul2006implementation}
Abu~S Abdul-Quader, Douglas~D Heckathorn, Keith Sabin, and Tobi Saidel.
\newblock Implementation and analysis of respondent driven sampling: lessons
  learned from the field.
\newblock {\em Journal of Urban Health}, 83(1):1--5, 2006.

\bibitem{gile_respondent-driven_2010}
Krista~J. Gile and Mark~S. Handcock.
\newblock Respondent-{Driven} {Sampling}: {An} {Assessment} of {Current}
  {Methodology}.
\newblock {\em Sociological Methodology}, 40(1):285--327, 2010.

\bibitem{gile_diagnostics_2015}
Krista~J. Gile, Lisa~G. Johnston, and Matthew~J. Salganik.
\newblock Diagnostics for respondent-driven sampling.
\newblock {\em Journal of the Royal Statistical Society: Series A (Statistics
  in Society)}, 178(1):241--269, January 2015.

\bibitem{mouw_network_2012}
Ted Mouw and Ashton~M. Verdery.
\newblock Network {Sampling} with {Memory} {A} {Proposal} for {More}
  {Efficient} {Sampling} from {Social} {Networks}.
\newblock {\em Sociological Methodology}, 42(1):206--256, 2012.

\bibitem{shi_model-based_2016}
Yongren Shi, Christopher~J. Cameron, and Douglas~D. Heckathorn.
\newblock Model-{Based} and {Design}-{Based} {Inference}: {Reducing} {Bias}
  {Due} to {Differential} {Recruitment} in {Respondent}-{Driven} {Sampling}.
\newblock {\em Sociological Methods \& Research}, page 0049124116672682,
  October 2016.

\bibitem{verdery_network_2015}
Ashton~M. Verdery, Ted Mouw, Shawn Bauldry, and Peter~J. Mucha.
\newblock Network structure and biased variance estimation in respondent driven
  sampling.
\newblock {\em PloS one}, 10(12):e0145296, 2015.

\bibitem{wejnert_empirical_2009}
Cyprian Wejnert.
\newblock An {Empirical} {Test} of {Respondent}-{Driven} {Sampling}: {Point}
  {Estimates}, {Variance}, {Degree} {Measures}, and {Out}-of-{Equilibrium}
  {Data}.
\newblock {\em Sociological Methodology}, 39(1):73--116, 2009.

\bibitem{heckathorn_network_2017}
Douglas~D. Heckathorn and Christopher~J. Cameron.
\newblock Network {Sampling}.
\newblock {\em Annual Review of Sociology}, 43(1), 2017.

\bibitem{sulaberidze_population_2016}
Lela Sulaberidze, Ali Mirzazadeh, Ivdity Chikovani, Natia Shengelia, Nino
  Tsereteli, and George Gotsadze.
\newblock Population {Size} {Estimation} of {Men} {Who} {Have} {Sex} with {Men}
  in {Tbilisi}, {Georgia}; {Multiple} {Methods} and {Triangulation} of
  {Findings}.
\newblock {\em PLoS ONE}, 11(2):e0147413, February 2016.

\bibitem{domingo-salvany_analytical_1998}
Antonia Domingo-Salvany, Richard~L. Hartnoll, Andrew Maguire, M.~Teresa Brugal,
  Pilar~Albertin Albertin, Joan~A. Caylà, Jordi Casabona, and Josep~M.
  Suelves.
\newblock Analytical considerations in the use of capture-recapture to estimate
  prevalence: case studies of the estimation of opiate use in the metropolitan
  area of {Barcelona}, {Spain}.
\newblock {\em American journal of epidemiology}, 148(8):732--740, 1998.

\bibitem{kruse_participatory_2003}
Natalie Kruse, M.-TF~Behets Frieda, Georgine Vaovola, Gillian Burkhardt, Texina
  Barivelo, X.~Amida, and Gina Dallabetta.
\newblock Participatory mapping of sex trade and enumeration of sex workers
  using capture-recapture methodology in {Diego}-{Suarez}, {Madagascar}.
\newblock {\em Sexually transmitted diseases}, 30(8):664--670, 2003.

\bibitem{bernard_counting_2010}
H.~R. Bernard, T.~Hallett, A.~Iovita, E.~C. Johnsen, R.~Lyerla, C.~McCarty,
  M.~Mahy, M.~J. Salganik, T.~Saliuk, O.~Scutelniciuc, G.~A. Shelley,
  P.~Sirinirund, S.~Weir, and D.~F. Stroup.
\newblock Counting hard-to-count populations: the network scale-up method for
  public health.
\newblock {\em Sexually Transmitted Infections}, 86(Suppl 2):ii11--ii15,
  November 2010.

\bibitem{hay_estimating_1996}
Gordon Hay and Neil McKeganey.
\newblock Estimating the prevalence of drug misuse in {Dundee}, {Scotland}: an
  application of capture-recapture methods.
\newblock {\em Journal of Epidemiology and Community Health}, 50(4):469--472,
  1996.

\bibitem{jones_recapture_2014}
Hayley~E. Jones, Matthew Hickman, Nicky~J. Welton, Daniela De~Angelis, Ross~J.
  Harris, and A.~E. Ades.
\newblock Recapture or {Precapture}? {Fallibility} of {Standard}
  {Capture}-{Recapture} {Methods} in the {Presence} of {Referrals} {Between}
  {Sources}.
\newblock {\em American journal of epidemiology}, page kwu056, 2014.

\bibitem{wolitski_effects_2009}
Richard~J. Wolitski, Sherri~L. Pals, Daniel~P. Kidder, Cari Courtenay-Quirk,
  and David~R. Holtgrave.
\newblock The effects of {HIV} stigma on health, disclosure of {HIV} status,
  and risk behavior of homeless and unstably housed persons living with {HIV}.
\newblock {\em AIDS and Behavior}, 13(6):1222--1232, December 2009.
\newblock 00049.

\bibitem{ezoe_population_2012}
Satoshi Ezoe, Takeo Morooka, Tatsuya Noda, Miriam~Lewis Sabin, and Soichi
  Koike.
\newblock Population size estimation of men who have sex with men through the
  network scale-up method in {Japan}.
\newblock {\em PloS one}, 7(1):e31184, 2012.

\bibitem{guo_estimating_2013}
Wei Guo, Shuilian Bao, Wen Lin, Guohui Wu, Wei Zhang, Wolfgang Hladik, Abu
  Abdul-Quader, Marc Bulterys, Serena Fuller, and Lu~Wang.
\newblock Estimating the size of {HIV} key affected populations in {Chongqing},
  {China}, using the network scale-up method.
\newblock {\em PloS one}, 8(8):e71796, 2013.

\bibitem{habecker_improving_2015}
Patrick Habecker, Kirk Dombrowski, and Bilal Khan.
\newblock Improving the {Network} {Scale}-{Up} {Estimator}: {Incorporating}
  {Means} of {Sums}, {Recursive} {Back} {Estimation}, and {Sampling} {Weights}.
\newblock {\em PloS one}, 10(12), 2015.

\bibitem{killworth_investigating_2006}
Peter~D. Killworth, Christopher McCarty, Eugene~C. Johnsen, H.~Russell Bernard,
  and Gene~A. Shelley.
\newblock Investigating the variation of personal network size under unknown
  error conditions.
\newblock {\em Sociological Methods \& Research}, 35(1):84--112, 2006.

\bibitem{salganik_assessing_2011}
Matthew~J. Salganik, Dimitri Fazito, Neilane Bertoni, Alexandre~H. Abdo,
  Maeve~B. Mello, and Francisco~I. Bastos.
\newblock Assessing network scale-up estimates for groups most at risk of
  {HIV}/{AIDS}: evidence from a multiple-method study of heavy drug users in
  {Curitiba}, {Brazil}.
\newblock {\em American journal of epidemiology}, 174(10):1190--1196, 2011.

\bibitem{TELEFUNKEN2012}
Kirk Dombrowski, Bilal Khan, Travis Wendel, Katherine McLean, Evan Misshula,
  and Ric Curtis.
\newblock Estimating the size of the methamphetamine-using population in new
  york city using network sampling techniques.
\newblock {\em Advances in Applied Sociology}, 2(4):1 -- 20, 2012.

\bibitem{curtis_commercial_2008}
R.~Curtis, K.~Terry, M.~Dank, K.~Dombrowski, and B.~Khan.
\newblock The commercial sexual exploitation of children in {New} {York}
  {City}, {Volume} 1: {The} {CSEC} population in {New} {York} {City}: {Size},
  characteristics, and needs ({NCJ} {Publication} {No}. 225083). {Bureau} of
  {Justice} {Statistics}, {Washington}, {DC}.
\newblock {\em Final report submitted to the National Institute of Justice. New
  York, NY: Center for Court Innovation and John Jay College of Criminal
  Justice. Retrieved January}, 12:2012, 2008.
\newblock Cited by 0008.

\bibitem{wendel_dynamics_2011}
Travis Wendel, Bilal Khan, Kirk Dombrowski, Ric Curtis, Katherine McLean, Evan
  Misshula, Robert Riggs, and David Marshall~IV.
\newblock {\em Dynamics of {Methamphetamine} {Markets} in {New} {York} {City}:
  {Final} {Technical} {Report} to the {National} {Institute} of {Justice}; {A}
  {Report} to the {National} {Institute} of {Justice} ({Award} \#
  2007-{IJ}-{CX}-0110}, volume NIJ Document 236122.
\newblock 2011.
\newblock Cited by 0000.

\bibitem{merli_sampling_2016}
M.~Giovanna Merli, Ashton Verdery, Ted Mouw, and Jing Li.
\newblock Sampling migrants from their social networks: {The} demography and
  social organization of {Chinese} migrants in {Dar} es {Salaam}, {Tanzania}.
\newblock {\em Migration studies}, 4(2):182--214, 2016.

\bibitem{LINCOLN1930}
F.~C. Lincoln.
\newblock Calculating waterfowl abundance on the basis of banding returns.
\newblock {\em United States Department of Agriculture Circular}, 118:1 -- 4,
  1930.

\bibitem{PETERSON1896}
C.~Peterson Petersen.
\newblock The yearly immigration of young plaice into the limfjord from the
  german sea.
\newblock {\em Report of the Danish Biological Station}, 6:5 -- 84, 1896.

\bibitem{CHAPMAN1951}
D.~G. Chapman.
\newblock Some properties of the hypergeometric distribution with applications
  to zoological sample censuses.
\newblock {\em UC Publications in Statistics}, 1(7):1 -- 159, 1952.

\bibitem{RDS1997}
Douglas~D. Heckathorn.
\newblock Respondent-driven sampling: A new approach to the study of hidden
  populations.
\newblock {\em Social Problems}, 1997.

\bibitem{Bollobas98a}
Bela Bollobas.
\newblock {\em Modern Graph Theory}.
\newblock Springer, 1998.

\bibitem{CARTER1979}
J.Lawrence Carter and Mark~N. Wegman.
\newblock Universal classes of hash functions.
\newblock {\em Journal of Computer and System Sciences}, 18(2):143 -- 154,
  1979.

\bibitem{STIRLING1988}
Ronald~L. Graham, Donald~E. Knuth, and Oren Patashnik.
\newblock {\em Concrete Mathematics}, volume~2.
\newblock Addison–Wesley, Reading MA, 1988.

\bibitem{o2003jung}
Joshua O?Madadhain, Danyel Fisher, Scott White, and Y~Boey.
\newblock The jung (java universal network/graph) framework.
\newblock {\em University of California, Irvine, California}, 2003.

\bibitem{CLUSTERING1971}
P.~W. Holland and S.~Leinhardt.
\newblock Transitivity in structural models of small groups.
\newblock {\em Comparative Group Studies}, 2:107 -- 124, 1971.

\end{thebibliography}


\begin{thebibliography}{10}

\bibitem{LINCOLN1930}
F.~C. Lincoln.
\newblock Calculating waterfowl abundance on the basis of banding returns.
\newblock {\em United States Department of Agriculture Circular}, 118:1 -- 4,
  1930.

\bibitem{PETERSON1896}
C.~Peterson Petersen.
\newblock The yearly immigration of young plaice into the limfjord from the
  german sea.
\newblock {\em Report of the Danish Biological Station}, 6:5 -- 84, 1896.

\bibitem{CHAPMAN1951}
D.~G. Chapman.
\newblock Some properties of the hypergeometric distribution with applications
  to zoological sample censuses.
\newblock {\em UC Publications in Statistics}, 1(7):1 -- 159, 1952.

\bibitem{RDS1997}
Douglas~D. Heckathorn.
\newblock Respondent-driven sampling: A new approach to the study of hidden
  populations.
\newblock {\em Social Problems}, 1997.

\bibitem{RDS2002}
Douglas~D. Heckathorn.
\newblock Respondent-driven sampling ii: Deriving valid population estimates
  from chain-referral samples of hidden populations.
\newblock {\em Social Problems}, 2002.

\bibitem{Bollobas98a}
Bela Bollobas.
\newblock {\em Modern Graph Theory}.
\newblock Springer, 1998.

\bibitem{CARTER1979}
J.Lawrence Carter and Mark~N. Wegman.
\newblock Universal classes of hash functions.
\newblock {\em Journal of Computer and System Sciences}, 18(2):143 -- 154,
  1979.

\bibitem{TELEFUNKEN2012}
Kirk Dombrowski, Bilal Khan, Travis Wendel, Katherine McLean, Evan Misshula,
  and Ric Curtis.
\newblock Estimating the size of the methamphetamine-using population in new
  york city using network sampling techniques.
\newblock {\em Advances in Applied Sociology}, 2(4):1 -- 20, 2012.

\bibitem{STIRLING1988}
Ronald~L. Graham, Donald~E. Knuth, and Oren Patashnik.
\newblock {\em Concrete Mathematics}, volume~2.
\newblock Addison–Wesley, Reading MA, 1988.

\bibitem{CLUSTERING1971}
P.~W. Holland and S.~Leinhardt.
\newblock Transitivity in structural models of small groups.
\newblock {\em Comparative Group Studies}, 2:107 -- 124, 1971.

\end{thebibliography}

\end{document}